%% file: main.tex
\documentclass{article}
\usepackage[top=1in, bottom=1in, left=1in, right=1in]{geometry}
\usepackage[utf8]{inputenc}
\usepackage{amsmath,amsthm,amsfonts,amssymb}
\usepackage{float}
\usepackage{url}
\usepackage{mathtools}
\usepackage{hyperref}
\usepackage[capitalise]{cleveref}
\usepackage{mathrsfs}
\usepackage{xcolor}
\usepackage{braket}

\usepackage[round]{natbib} 
\usepackage{graphicx} 
\usepackage{tabularx}
\usepackage{comment}
\usepackage{mleftright}
\mleftright
\usepackage{makecell}

\newcommand{\E}{\mathbb{E}}

\DeclareMathOperator{\spec}{spec}

\newcommand{\lind}{{\mathcal{L}}}

\newcommand{\trace}[1]{\operatorname{tr}(#1)}

\newtheorem{theorem}{Theorem}[section]

\newtheorem*{theorem*}{Theorem}
\newtheorem{lemma}[theorem]{Lemma}
\newtheorem*{lemma*}{Lemma}
\newtheorem{proposition}[theorem]{Proposition}
\newtheorem*{prop*}{Proposition}
\newtheorem{corollary}[theorem]{Corollary}
\newtheorem{remark}[theorem]{Remark}

\newtheorem{claim}[theorem]{Claim}
\newtheorem{definition}[theorem]{Definition}

\usepackage{comment}

\makeatletter
\let\c@equation\c@theorem
\makeatother

\DeclareMathOperator*{\argmax}{arg\,max}

\renewcommand{\braket}[2]{\langle#1|#2\rangle}

\input{macros}
\title{On quantum to classical comparison for Davies generators}
\author{
Joao Basso\thanks{Authors are listed in alphabetical order. Department of Mathematics, University of California, Berkeley, CA 94720, USA. \texttt{joao.basso@berkeley.edu}} 
\and Shirshendu Ganguly\thanks{Department of Statistics, University of California, Berkeley, CA 94720, USA. \texttt{sganguly@berkeley.edu}} 
\and Alistair Sinclair\thanks{Computer Science Division, University of California, Berkeley, CA 94720, USA. \texttt{sinclair@cs.berkeley.edu}}  
\and Nikhil Srivastava\thanks{Department of Mathematics, University of California, Berkeley, CA 94720, USA. \texttt{nikhil@math.berkeley.edu}} 
\and Zachary Stier\thanks{Department of Mathematics, University of California, Berkeley CA 94720, USA. \texttt{zstier@berkeley.edu}} 
\and Thuy-Duong~Vuong\thanks{Computer Science Division, University of California, Berkeley, CA 94720, USA, and Miller Institute for Basic Science Research. \texttt{tdvuong@berkeley.edu}. \\
Department of Computer Science and Engineering, University of California, San Diego, CA 92093, USA. \texttt{thvuong@ucsd.edu} (preferred email)}
}

\begin{document}
\maketitle
\begin{abstract}
Despite extensive study, our understanding of quantum Markov chains remains far less complete than that of their classical counterparts. Temme \cite{temme13} observed that the Davies Lindbladian, a well-studied model of quantum Markov dynamics, contains an embedded classical Markov generator, raising the natural question of how the convergence properties of the quantum and classical dynamics are related. While \cite{temme13} showed that the spectral gap of the Davies Lindbladian can be much smaller than that of the embedded classical generator for certain highly structured Hamiltonians, we show that if the spectrum of the Hamiltonian does not contain long arithmetic progressions, then the two spectral gaps must be comparable. As a consequence, we prove that for a large class of Hamiltonians, including those obtained by perturbing a fixed Hamiltonian with a generic external field, the quantum spectral gap remains within a constant factor of the classical spectral gap. Our result aligns with physical intuition and enables the application of classical Markov chain techniques to the quantum setting.

The proof is based on showing that any ``off-diagonal'' eigenvector of the Davies generator can be used to construct an observable which commutes with the Hamiltonian and has a Lindbladian Rayleigh quotient which can be upper bounded in terms of that of the original eigenvector's Lindbladian Rayleigh quotient. Thus, a spectral gap for such observables implies a spectral gap for the full Davies generator. 
\end{abstract}
\input{intro}
\input{prelim}
\input{general_comparison_theorem}
\input{spectral_gap_in_V0_vs_classical_chain}
\input{absence_of_long_arithmetic_progression}

\section*{Acknowledgements}
We thank Thiago Bergamaschi, Lin Lin, Michael Ragone, Kevin Stubbs, and James Sud for insightful discussions.
JB, NS, and ZS are funded by NSF Grant CCF-2420130. ZS was also supported by NSF Grant DGE-2146752. AS was funded in part by NSF grant CCF-2231095. TDV is supported by the Miller Institute for Basic research in Science. SG is supported in part by NSF Career grant 1945172 and the Miller Institute of Basic research in Science.

\bibliographystyle{plainnat}  
\bibliography{main}
\input{appendix}
\end{document}

%% file: macros.tex
\newcommand{\C}[0]{\mathbb{C}}

\newcommand{\N}[0]{\mathbb{N}}

\newcommand{\R}[0]{\mathbb{R}}

\newcommand{\Z}[0]{\mathbb{Z}}

\newcommand{\tr}[0]{\operatorname{tr}}

\newcommand{\Var}[0]{\operatorname{Var}}

\newcommand{\pull}[9]{
#1\ar@/_/[ddr]_{#2} \ar@{.>}[rd]^{#3} \ar@/^/[rrd]^{#4} & &\\
& #5\ar[r]^{#6}\ar[d]^{#8} &#7\ar[d]^{#9} \\}

\newcommand{\cmp}[9]{
\xymatrix{
#1 \ar[r]^{#4}{#5} \ar@/_2pc/[rr]^{#8}_{#9} & #2 \ar[r]^{#6}_{#7} & #3
}
}

\newcommand{\ha}[1]{\ar@{^(->}[#1]}
\newcommand{\ls}[1]{\ar@{-}[#1]}
\newcommand{\sj}[1]{\ar@{->>}[#1]}
\newcommand{\aq}[1]{\ar@{=}[#1]}
\newcommand{\acir}[1]{\ar@{}[#1]|-{\textstyle{\circlearrowright}}}
\newcommand{\acil}[1]{\ar@{}[#1]|-{\textstyle{\circlearrowleft}}}
\newcommand{\ard}[1]{\ar@{.>}[#1]}
\newcommand{\mt}[1]{\ar@{|->}[#1]}
\newcommand{\inm}[1]{\ar@{}[#1]|-{\in}}
\newcommand{\inr}{\ar@{}[d]|-{\rotatebox[origin=c]{-90}{$\in$}}}
\newcommand{\inl}{\ar@{}[u]|-{\rotatebox[origin=c]{90}{$\in$}}}

\newcommand{\beq}[1]{\begin{equation}\llabel{#1}}
\newcommand{\eeq}[0]{\end{equation}}
\newcommand{\bal}[0]{\begin{align*}}
\newcommand{\eal}[0]{\end{align*}}
\newcommand{\ban}[0]{\begin{align}}
\newcommand{\ean}[0]{\end{align}}

\newcommand{\fixme}[1]{{\color{red}#1}}
\newcommand{\llabel}[1]{\label{#1}\text{\fixme{\tiny#1}}}

\newcommand{\arxiv}[1]{\url{http://www.arxiv.org/abs/#1}}

\allowdisplaybreaks[2]

\DeclareFontFamily{U}{wncy}{}
    \DeclareFontShape{U}{wncy}{m}{n}{<->wncyr10}{}
    \DeclareSymbolFont{mcy}{U}{wncy}{m}{n}
    \DeclareMathSymbol{\Sh}{\mathord}{mcy}{"58} 

%% file: intro.tex
\section{Introduction}
For a Hamiltonian $H$ and inverse temperature parameter $\beta\geq 0$, the goal of  quantum Gibbs sampling is to prepare
the quantum Gibbs state \[\rho_\beta := \frac{\exp(-\beta H)}{\trace{\exp(-\beta H)}} .\]
One of the leading approaches for quantum Gibbs sampling is based on quantum Markov processes, with the Davies Lindbladian generator being a prominent model (see Definition \ref{def:davies generator}). Despite significant recent progress, our understanding of the Davies generator remains incomplete. Existing works are able to prove mixing for commuting Hamiltonians \cite{Temme2014ThermalizationTB,kastoryano2016quantum,kochanowski2024rapid}, in the high temperature regime \cite{rouze2025efficient}, or under strong assumptions on either the Hamiltonian or the Davies generator \cite{Bardet2021EntropyDF,rajakumar2024gibbs,chen2021fast,chen2024randomized}. Yet, there is a lack of tools for proving mixing outside these settings.
This stands in stark contrast to the classical setting, where recent advances have yielded sharp convergence bounds across the full range of temperatures for several important models \cite{ALO20,CLV21,AJKPV22,Chen2022LocalizationSA,Chen2022OptimalMF}.

A natural question arises: to what extent can techniques from classical Markov chains be leveraged to understand their quantum counterparts? In \cite{temme13}, it is observed that when $H$ has simple spectrum, the Davies Lindbladian contains a canonical embedded classical Markov chain. Specifically, its action on observables that are diagonal in the Hamiltonian eigenbasis is isomorphic to that of a classical Markov generator called the Pauli master equation. This implies that the spectral gap $\lambda_{Q}$ of the quantum process is upper-bounded by the spectral gap $\lambda_{\mathrm{cl}}$ of the embedded classical one, i.e., $\lambda_{Q}\le \lambda_{\mathrm{cl}}$. 

Temme asked the insightful question: when does the opposite comparison $\lambda_Q \gtrsim^? \lambda_{\mathrm{cl}}$ hold? He gave two kinds of answers to this question.  First, he 
constructed a single-body Hamiltonian with simple spectrum and $\lambda_Q=O(1/N)\lambda_{\mathrm{cl}}$, where $N$ is
the Hilbert space dimension, showing that an efficient comparison is impossible in general. Second, he outlined
an approach to proving bounds on $\lambda_Q$ on a case-by-case basis by considering the matrix entries of the Davies generator (in a basis induced by the Hamiltonian eigenbasis)
and applying matrix-analytic tools such as Gershgorin's theorem to control the off-diagonal part (i.e., its action orthogonal to the diagonal). Since then, this approach was used in several works (e.g. \cite{chen2021fast,rajakumar2024gibbs,chen2024randomized}) to separately control the diagonal and off-diagonal parts of a Davies generator, and conclude that in those cases $\lambda_Q$ is not much smaller than $\lambda_{\mathrm{cl}}$. Temme concluded his paper by asking whether his approach could be extended to Hamiltonians with degenerate spectrum, which many interesting examples such as classical Ising models, Toric code Hamiltonians, and Heisenberg models exhibit.

In this paper, we address Temme's question by identifying a condition on the spectrum of $H$ which we call the {\em short arithmetic progression property}, which implies that the quantum and classical convergence rates of the Davies generator are comparable. Our main result (\cref{thm:general comparison theorem dependent only on length of arithmetic sequence}) shows that if the length of the longest {arithmetic progression} in $\spec(H)$ is $D$, then $\lambda_Q\ge\Omega(1/D)\lambda_{\mathrm{cl}}$. This is consistent with Temme's work, as his example was a Hamiltonian whose spectrum contained a long arithmetic progression of length $N$; our result thus shows that long arithmetic progressions are the only obstacle to such a comparison. Notably, the short AP property we introduce is considerably weaker than having simple spectrum and simple Bohr spectrum\footnote{The simple Bohr spectrum condition means that for any eigenvalues \( E_j \neq E_k \) and \( E_{j'} \neq E_{k'} \) of \( H \), the corresponding energy differences are distinct; that is, \( E_j - E_k \neq E_{j'} - E_{k'} \).} (see Remark \ref{rem:bohrvsap}), and includes many natural examples which have degeneracy, including the ones mentioned above (see \cref{rem:proper vs non-proper AP}).

Secondarily, we show that several families of generic Hamiltonians --- notably, any many-body Hamiltonian with a random local external field (\cref{thm:generic transverse field}) --- obey the short AP property with $D=2$, implying that for such models the convergence rates of the quantum and classical processes are within a constant factor of each other. 
As an interesting consequence (\cref{cor:cheeger}), by applying the classical Cheeger inequality we conclude that for short AP Hamiltonians the Davies Lindbladian spectral gap is approximately witnessed by a projection which commutes with the Hamiltonian.

Before stating our results in detail, we recall the definition of the Davies generator and record some preliminary facts.
\begin{definition}[Davies generator] \label{def:davies generator} Let $n$ be the number of qubits and $N = 2^n$ be the dimension of the corresponding Hilbert space.
    Given a Hamiltonian $H\in \C^{N\times N},$ an inverse temperature $\beta\geq 0\footnote{Our results hold for all $\beta \in \mathbb{R}$, but we assume $\beta \geq 0$ to align with standard convention.},$ a set of jump operators $\mathcal{S},$ and a transition rate function $G:\R\rightarrow\R,$ the Davies generator\footnote{The general formulation of the Davies generator includes a coherent term, writing the generator as 
\[
\tilde{\mathcal{L}}(f) = i[H, f] + \mathcal{L}(f),
\]
where \( \mathcal{L} \) is as \cref{def:davies generator}. This is important for detailed balance \cite{chen2023efficient} and hypocoercivity arguments \cite{fang2025mixing}, but the inclusion of the coherent term $i[H,f]$ does not change the spectral gap so we will omit it in this paper.} is a linear operator $\mathcal{L}$ on $\C^{N\times N}$ defined for observable $f$ as
    \begin{equation}\label{def:davies dissipative part}
       \mathcal{L} (f) := \sum_{\omega, S\in \mathcal{S}}\mathcal{L}_{\omega, S}(f),\nonumber 
    \end{equation}
     where 
for Bohr frequency $\omega$ and jump operator $S \in \mathcal{S},$
    \begin{equation}
        \mathcal{L}_{\omega, S}(f) := G(\omega) (S(\omega)^\dag f S(\omega) - \frac{1}{2} \{S(\omega)^\dag S(\omega), f\} )= \frac{G(\omega)}{2} \left(S(\omega)^\dag [f, S(\omega) ] + [S(\omega)^\dag, f] S(\omega)\right),\nonumber
    \end{equation}
    where $[A,B]=AB-BA$ denotes the commutator, $\{A,B\} = AB + BA$ denotes the anticommutator, and \[S(\omega) := \sum_{\lambda_1, \lambda_2: \lambda_1 - \lambda_2 = \omega} \Pi_{\lambda_1} S \Pi_{\lambda_2},\]
where $ \Pi_\lambda$ is the projector to the eigenspace of $H$ of eigenvalue $\lambda.$
\end{definition}

Throughout the paper, we make the standard assumption that:
\begin{itemize}
    \item  The jump operator set is self-adjoint,
\begin{equation} \label{assumption:jump operator closure}
    \mathcal{S}^\dag = \set{S^\dag : S \in \mathcal{S}} = \mathcal{S},\nonumber
    \end{equation}
    \item The transition rate function satisfies the detailed balance condition: $ G(\omega)= G(-\omega) e^{-\beta \omega}.$ 
\end{itemize}

Under these assumptions, 
  $\mathcal{L}$ satisfies the following Kubo--Martin--Schwinger (KMS) reversibility condition (see \cite[Fact 3.1]{BCL24}), which can be viewed as an analog of the time-reversibility condition of classical Markov chains:
\begin{equation}\label{eq:kms reversible condition}
    \langle \mathcal{L}(f), g\rangle_\rho = \langle f, \mathcal{L}(g)\rangle_\rho,
\end{equation}
where recall that $\rho \equiv  \rho_\beta =\frac{\exp(-\beta H)}{\trace{\exp(-\beta H)}}$ is the Gibbs state, and the inner product $\langle \cdot, \cdot \rangle_\rho$ is the KMS inner product, i.e.\ $ \langle A, B\rangle_{\rho} := \tr( \rho^{1/2} A \rho^{1/2} B^\dag). $ 

For an observable $f$, define its
\begin{itemize}
    \item Variance: $\Var(f):= \langle f, f\rangle_\rho - \trace{\rho f} \trace{\rho f^\dag}\nonumber$;
    \item Dirichlet form: $\mathcal{E}_{\mathcal{L}}(f):= - \langle \mathcal{L}(f), f\rangle_\rho$; and
    \item Rayleigh quotient of $f$: $ \lambda_{\mathcal{L}}(f):= \frac{   \mathcal{E}_{\mathcal{L}}(f)}{ \Var(f)}$.\footnote{with the convention that $\lambda_{\mathcal{L}}(f) = +\infty$ if $ \Var(f) = 0.$}
\end{itemize}
We define the spectral gap of $\mathcal{L}$ by
\begin{equation} \label{eq:spectral gap}
\lambda_{\mathcal{L}} :=\min_{f}\lambda_{\mathcal{L}}(f).
\end{equation}

The spectral gap of $\mathcal{L}$ controls the convergence rate of the Lindbladian dynamics $e^{t \mathcal{L}}$, specifically its $\chi^2$ mixing time (see \cite{KT13}). We will restrict our attention exclusively to the spectral gap in this paper, and will not discuss other approaches to establishing mixing. An important role is played by the following invariant subspaces of the Davies generator.


 \begin{definition}[$V_\omega$ subspaces]\label{def:V omega}
For a Hamiltonian $H$ and Bohr frequency $\omega,$ let $V_\omega$ be the image of the map \[f \mapsto f(\omega) := \sum_{\lambda_1, \lambda_2: \lambda_1 - \lambda_2 = \omega} \Pi_{\lambda_1} f \Pi_{\lambda_2}.\]
where $ \Pi_\lambda$ is the projector to eigenspace of $H$ of eigenvalue $\lambda.$ 
 \end{definition}
We note that $V_0$ is precisely the set of observables that commute with $H$ (see \cref{prop:observable in V0 commute with $H$}).

As shown in \cite{temme13}, each $V_\omega$ is an invariant subspace under $\mathcal{L}$, that is, $\lind(V_\omega) \subseteq V_\omega$. Define  the spectral gap in $V_\omega$ as
\begin{equation} \label{eq:spectral gap on subspace} \lambda_{\mathcal{L},\omega} = \min_{f\in V_\omega} \lambda_{\mathcal{L}}(f). 
\end{equation}
Then
\begin{equation}\label{eq:spectral gap is minimize over all subspace}
  \lambda_{\mathcal{L}}=\min_\omega  \lambda_{\mathcal{L}, \omega}.
\end{equation}

We also define the spectral gap with respect to Hermitian observables: 
\begin{equation} \label{eq:spectral gap hermitian} \lambda_{\mathcal{L}}^\mathrm{H} := \min_{f:f=f^\dag} \lambda_{\mathcal{L}}(f) \quad \text{and} \quad \lambda_{\mathcal{L}, 0}^\mathrm{H} := \min_{f\in V_0:f=f^\dag} \lambda_{\mathcal{L}}(f).
\end{equation}

In \cref{prop:compare spectral gap with hermitian spectral gap,prop:compare spectral gap with hermitian spectral gap in V0}, we show that
\(
\lambda_{\mathcal{L}}^\mathrm{H} = \lambda_{\mathcal{L}}\)  {and} \(\lambda_{\mathcal{L},0}^\mathrm{H} = \lambda_{\mathcal{L},0}.\)
When the Hamiltonian $H$ has simple spectrum,
\cite{temme13} relates the spectral gap $\lambda_{\mathcal{L}, 0}^\mathrm{H}$ of Hermitian observables in $V_0$ to the spectral gap of a certain classical Markov generator.
\begin{definition}[Classical Markov generator, {\cite{temme13}}] \label{def:classical Markov generator}
Consider a Hamiltonian \( H \) with an orthonormal eigenbasis \( U = \{ u_i \}_{i=1}^N \). Given a Davies generator \( \mathcal{L} \) associated with $H$, inverse temperature $ \beta \geq 0,$ transition rate function $G:\R\to \R,$ and jump operator set $\mathcal{S}$,  we define a corresponding classical Markov generator with state space \( U \), using the following construction.

Let $ \rho = \rho_\beta := \frac{\exp(-\beta H)}{\trace{\exp(-\beta H)}},$
and define the corresponding stationary distribution \( \pi: U \to \mathbb{R}_{\geq 0} \) by
\[
\pi(u_i) := \langle u_i | \rho | u_i \rangle.
\]
The transition rate of the classical Markov generator from state \( u_i \) to \( u_j \) (for \( i \neq j \)) is given by
\begin{equation}\label{eq:classical MC transition probability}
P_{\mathcal{L},U}[u_i \to u_j] = \langle \mathcal{L}(\ket{u_j}\bra{u_j}), \ket{u_i}\bra{u_i} \rangle = G(E_j - E_i) \sum_{S\in \mathcal{S}} |\langle u_i | S | u_j \rangle|^2 \geq 0,\nonumber
\end{equation}
where \( E_i = \langle u_i | H | u_i \rangle \) is the eigenvalue of $H$ corresponding to $u_i.$

Given a test function \( F: U \to \mathbb{R} \), 
the Dirichlet form of $F$ is defined as
\begin{equation}\label{eq:classical chain dirichlet form}
\mathcal{E}_{\mathcal{L}, U, \mathrm{cl}}(F) = \frac{1}{2} \sum_{i,j} (F(u_i) - F(u_j))^2 \, \pi(u_i) P_U[u_i \to u_j].\nonumber
\end{equation}

The \emph{classical spectral gap} associated with \( \mathcal{L} \) and the basis \( U \) is then given by
\begin{equation}\label{eq:classical chain spectral gap}
\lambda_{\mathcal{L}, U, \mathrm{cl}} := \min_{F: U \to \mathbb{R}} \frac{\mathcal{E}_{\mathcal{L}, U, \mathrm{cl}}(F)}{\mathrm{Var}_\pi(F)}\footnote{with the convention that $ \frac{\mathcal{E}_{\mathcal{L}, U, \mathrm{cl}}(F)}{\mathrm{Var}_\pi(F)} =+\infty$ if $\Var_\pi(F) = 0. $},\nonumber
\end{equation}
where \( \mathrm{Var}_\pi(F) := \mathbb{E}_\pi[F^2] - \mathbb{E}_\pi[F]^2 \) denotes the variance of \( F \) with respect to the distribution \( \pi \).
\end{definition}
The reversibility condition of the Lindbladian $\mathcal{L}$ with respect to the KMS inner product $\langle \cdot, \cdot \rangle_\rho$ implies that the associated classical Markov generators are reversible with respect to~$\pi$.
  
\begin{proposition}[{\cite{temme13}}] \label{prop:V0 equal classical for nondegenerate H}
    For a Hamiltonian $H$ that has simple spectrum and a Davies generator $\mathcal{L}$ associated with $H,$ let the \emph{classical spectral gap} of $\mathcal{L},$ denoted by $\lambda_{\mathcal{L},\mathrm{cl}}, $ be the spectral gap of the classical Markov generator associated with $\mathcal{L}$ and the unique orthonormal eigenbasis of $H$ (see \cref{def:classical Markov generator}). Then, \[ \lambda_{\mathcal{L},0}^\mathrm{H} = \lambda_{\mathcal{L},\mathrm{cl}}. \]
\end{proposition}

\subsection{Our results}


\paragraph{Comparison under the Short AP Property.}
Our main theorem is a comparison between the spectral gap of the Davies generator $\mathcal{L}$ and the spectral gap of observables in $V_0$.
Before stating the result, we introduce the following definition.

\begin{definition}[Arithmetic progression]
    A multiset \( S \) is said to contain an \emph{arithmetic progression of length \( D \)} (or a {\em \( D \)-arithmetic progression}) if there exist \( a, b  \) such that
    \[
    \{ a, a + b, a + 2b, \ldots, a + (D - 1)b \} \subseteq S.
    \]
    If \( b \neq 0 \), we say that \( S \) contains a \emph{proper} \( D \)-arithmetic progression.
\end{definition}

\begin{theorem}[General spectral gap comparison theorem] \label{thm:general comparison theorem dependent only on length of arithmetic sequence}
      If the spectrum of a Hamiltonian $H$ does not contain any proper $(D+1)$-arithmetic progression, then for any Davies generator $\mathcal{L}$ associated with $H,$
     \[ \lambda_{\mathcal{L}, 0} \geq \lambda_{\mathcal{L}} \geq \frac{1}{2D} \lambda_{\mathcal{L}, 0}.\]
\end{theorem}
\cref{thm:general comparison theorem dependent only on length of arithmetic sequence} is the consequence of \cref{eq:spectral gap is minimize over all subspace} and the following theorem, which compares the spectral gap in $V_\omega$ for $\omega \neq 0$ with the spectral gap  in $ V_0.$
\begin{theorem} \label{thm:compare spectral gap V omega and V 0}
    Consider a Hamiltonian $H$ and a Bohr frequency $\omega \neq 0.$ If the spectrum of $H$ does not contain a $(D+1)$-arithmetic progression where the difference between consecutive terms is $\omega$,
    then for any Davies generator $\mathcal{L}$ associated with $H,$
     \[\lambda_{\mathcal{L}, \omega}\geq \frac{1}{2}\max \left(\frac{1}{D}, 1 -e^{-|\beta \omega |} \right) \lambda_{\mathcal{L}, 0}.\]
\end{theorem}
When $H$ has simple spectrum, observables in $V_0$ are referred to as \emph{diagonal} observables, since they are diagonal in the unique eigenbasis of $H,$ while observables in $\bigcup_{\omega \neq 0} V_\omega$ are called \emph{non-diagonal} observables. In this case, the quantity $\min_{\omega \neq 0} \lambda_{\mathcal{L}, \omega}$ characterizes the decay rate of non-diagonal observables, and is known as the \emph{quantum decoherence rate}. \Cref{thm:compare spectral gap V omega and V 0} provides a comparison between the quantum decoherence rate and the decay rate of diagonal observables.

We also obtain the following generalization of \cite{temme13}, which relates the spectral gap in $V_0$ to the spectral gap of a certain classical Markov generator, even when $H$ has repeated eigenvalues. The key observation is that any Hermitian observables in $V_0$ is diagonal in some orthonormal eigenbasis of $H.$
\begin{proposition}\label{prop:classical chain V0}

  For any Hamiltonian $H$ and any Davies generator $\mathcal{L}$ associated with $H,$ there exists an orthonormal eigenbasis $\hat{U}$ of $H$ such that the spectral gap of (Hermitian or arbitrary) observables in $V_0$ is equal to the spectral gap of the classical Markov operator associated with $\mathcal{L}$ and $\hat{U}$, i.e.,
  \[\lambda_{\mathcal{L},0}  = \lambda_{\mathcal{L},0}^\mathrm{H} = \lambda_{\mathcal{L},\hat{U},\mathrm{cl}}.\]
  Moreover, $\hat{U}$ minimizes the associated classical spectral gap over all orthonormal eigenbases of $H$, i.e.,
\[
\hat{U} = \arg\min_{U} \lambda_{\mathcal{L},U,\mathrm{cl}}.
\]

\end{proposition}
The following is a direct corollary of \cref{thm:general comparison theorem dependent only on length of arithmetic sequence} and \cref{prop:classical chain V0}.
\begin{theorem}\label{thm:quantum vs classical spectral gap comparison}
    If the spectrum of a Hamiltonian $H$ does not contain a proper $(D+1)$-arithmetic progression, then for any Davies generator $\mathcal{L},$ there exists an orthonormal eigenbasis
    $\hat{U}$
    of $H$ such that the spectral gap of $\mathcal{L}$ can be compared to the spectral gap $\lambda_{\mathcal{L},\hat{U},\mathrm{cl}}$ of the classical Markov operator associated with $\mathcal{L}$ and $\hat{U}$ via the following sandwiching inequality:
    \[ \lambda_{\mathcal{L},\hat{U},\mathrm{cl}} \geq \lambda_{\mathcal{L}} \geq \frac{1}{2D} \lambda_{\mathcal{L},\hat{U},\mathrm{cl}} .\]
   Moreover, $\hat{U}$ minimizes the classical spectral gap over all orthonormal eigenbases of $H$, i.e.,
\[
\hat{U} = \arg\min_{U} \lambda_{\mathcal{L},U,\mathrm{cl}}.
\]
If additionally, the Hamiltonian $H$ has a simple spectrum, then $H$ has a unique orthonormal eigenbasis, and the quantum spectral gap $\lambda_{\mathcal{L}}$ can be compared to the spectral gap of the Markov generator associated with this unique eigenbasis, that is, the classical spectral gap $ \lambda_{\mathcal{L},\mathrm{cl}}$ of $\mathcal{L},$ via
\[ \lambda_{\mathcal{L},\mathrm{cl}} \geq \lambda_{\mathcal{L}} \geq \frac{1}{2D} \lambda_{\mathcal{L},\mathrm{cl}}. \]
\end{theorem}

Another corollary of \cref{thm:general comparison theorem dependent only on length of arithmetic sequence} and \cref{prop:classical chain V0} is that for short AP Hamiltonians, the spectral gap of a Davies Lindbladian is approximately witnessed by a projection which commutes with the Hamiltonian. The proof is to apply the classical Cheeger inequality to the classical Markov generator associated with $\mathcal{L}$ and $\hat{U}$.

\begin{corollary}\label{cor:cheeger}
    Suppose the spectrum of $H$ does not contain a proper arithmetic progression of length $D+1$. Let $\mathcal{L}$ be a Davies generator associated with $H$, inverse temperature $ \beta \geq 0,$ transition rate function $G:\R\to \R,$ and jump operator set $\mathcal{S}$. Then, there is a projection $P$ (i.e. $P^2 = P$) such that $[P,H]=0$ and
    $$ \mathcal{E}_\mathcal{L}(P)\le 4 \sqrt{D M \lambda_{\mathcal{L}}} \cdot \Var_\rho(P),$$
    where \[M = \|G\|_\infty \cdot \|\sum_{S\in \mathcal{S}} S S^\dag \| ,\]
   and $\| \cdot\|$ denotes the matrix spectral norm and $\|G\|_\infty := \sup_\omega |G(\omega)|$.
\end{corollary}
The factor \(M\) comes from \(\max_u \sum_{v \neq u} P_{\mathcal{L}, \hat{U}}[u \to v]\), i.e., the maximum total transition rate out of any state of the associated classical Markov operator. For common choices of \(G\), such as \(G(\omega) = \frac{1}{1 + \exp(\beta \omega)}\) and \(G(\omega) = \min\big(1, \exp(-\frac{\beta \omega}{2})\big)\), which corresponds to the transition rates Glauber dynamics and Metropolis-Hasting dynamics respectively, we have \(\|G\|_\infty \leq 1\), and with jump operators normalized so that \(\sum_{S \in \mathcal{S}} S S^\dagger \leq I\), it follows that \(M\leq 1.\)

\paragraph{Short AP Property for Generic Hamiltonians.}
Next, we show that the absence of 3-term arithmetic progressions is a \emph{generic} property under various models of perturbation. For instance, we prove that adding a generic external field to any fixed Hamiltonian typically yields a new Hamiltonian whose spectrum contains neither 3-term arithmetic progressions nor repeated eigenvalues\footnote{A special case was considered in \cite{Merkli2008DynamicsOC}, focusing on the restricted setting \( H = H_0 + \sum_i h_i Z_i \) with \( H_0 = \sum_{i,j} J_{ij} Z_i Z_j \); our result applies more generally to arbitrary \( H_0 \) and to broader classes of perturbation (see \cref{thm:generic transverse field} and \cref{sec:no AP in perturbed Hamiltonian}).}.

\begin{theorem}\label{thm:generic transverse field}
For a fixed Hamiltonian $H_0\in \C^{N\times N}$ and Hermitian $P\in \C^{2\times 2}$ where $P$ is not a multiple of the identity, let $P_i := I^{\otimes (i-1)} \otimes P \otimes I^{\otimes (n - i)}$, and let $H := H_0 + \sum_{i=1}^n h_i P_i$ where $\mathbf{h} = (h_i)_{i=1}^n\in \R^n$. Then the set of $\mathbf{h}$ such that $ H$ has a $3$-arithmetic progression or repeated eigenvalues has Lebesgue measure $0$.
\end{theorem}

Our technique also applies to other types of perturbations (see \cref{sec:no AP in perturbed Hamiltonian}). 
A direct consequence of our main comparison theorems (\cref{thm:general comparison theorem dependent only on length of arithmetic sequence,thm:quantum vs classical spectral gap comparison}), together with the absence of long proper arithmetic progressions in the spectra of generic Hamiltonians (\cref{thm:generic transverse field,thm:generic higher order term,thm:generic quadratic term,thm:1D XY+Z}), is that for such Hamiltonians, the quantum spectral gap is within a constant factor of the spectral gap of a corresponding classical Markov generator. In particular, we obtain the following:

\begin{corollary}\label{cor:generic transverse field}
For a fixed Hamiltonian $H_0\in \C^{N\times N}$ and Hermitian $P\in \C^{2\times 2}$ where $P$ is not a multiple of the identity, let $P_i := I^{\otimes (i-1)} \otimes P \otimes I^{\otimes (n - i)}$, and let $H := H_0 + \sum_{i=1}^n h_i P_i$ where $\mathbf{h} = (h_i)_{i=1}^n\in \R^n$. Then, for Lebesgue almost every $\mathbf{h}$, the Hamiltonian $H$ has simple spectrum and for every  Davies generator $\mathcal{L}$ associated with $H$:
$
\lambda_{\mathcal{L}} = \Theta(1) \lambda_{\mathcal{L},\mathrm{cl}}.
$
\end{corollary}
\begin{remark}\label{rem:bohrvsap}
The assumption that the spectrum of $H$ contains no long proper arithmetic progressions is related to, but significantly weaker than, the condition that $H$ has simple spectrum and simple Bohr spectrum--also known as the non-degenerate energy gaps condition (see, e.g., \cite[Definition 7]{HuangHarrow25}).

Clearly, when $H $ has simple spectrum and simple Bohr spectrum, \cref{thm:general comparison theorem dependent only on length of arithmetic sequence} applies with $D =2.$
Moreover, under this assumption, each subspace $V_\omega$ with $\omega \neq 0$ has dimension 1. In that case, it is easy to see that $ 
\lambda_{\mathcal{L}} \geq \frac{1}{2} \lambda_{\mathcal{L}, \mathrm{cl}}$  (see \cref{prop:simple comparison result when orbit has dimension 1}).

On the other hand, the assumption that $H$ does not contain long (proper) arithmetic progressions allows for a much richer structure. In particular, it permits the subspaces $V_\omega$ (for $\omega \neq 0$) to have dimension as large as $\Omega(2^n)$.

As an example, consider $H = \sum_{i=1}^n h_i Z_i$, or the 1D XY+Z model with periodic boundary conditions (see e.g.\ \cite{Keating_2015}):
\[
 H  =  J \sum_{i=1}^n X_i Y_{i+1} +  h \sum_{i=1}^n Z_i \quad \text{where}\quad Y_{n+1}= Y_1.
\]
The spectra of these Hamiltonians generically do not contain 3-term arithmetic progressions or repeated eigenvalues (see the proofs of \cref{thm:generic transverse field,thm:1D XY+Z}). The eigenvalues of $H$ in these cases take the form
\[
\lambda_{\textbf{z}} = \sum_{i=1}^n z_i m_i, \quad \text{with } \textbf{z} \in \{\pm 1\}^n,
\]
where the \( m_i \) are fixed parameters.

Take $\omega = 2m_n$. Then for any pair of configurations $\textbf{z}, \textbf{z'}$ satisfying $z_i = z'_i$ for all $i \in [n-1]$ and $z_n = -z'_n$, we have:
\[
\lambda_{\textbf{z}} - \lambda_{\textbf{z}'} = \omega.
\]
There are $2^{n-1}$ such pairs, implying that the dimension of $V_\omega$ is at least $2^{n-1}$. This illustrates that the assumption of avoiding long arithmetic progressions still allows $V_\omega$ with $\omega \neq 0$ to have exponentially large dimension and complex structure.
\end{remark}

\begin{remark} \label{rem:proper vs non-proper AP}
It is important to distinguish between \emph{proper} and \emph{non-proper} arithmetic progressions when applying \cref{thm:general comparison theorem dependent only on length of arithmetic sequence}.

For instance, consider the Hamiltonian \( H = \sum_{i,j} Z_i Z_j \). This Hamiltonian has only \( O(n) \) distinct eigenvalues, so by the pigeonhole principle, at least one eigenvalue must have multiplicity \( 2^{n - o(1)} \). As a result, the spectrum of \( H \) contains a \emph{non-proper} arithmetic progression of length \( 2^{n - o(1)} \), where all terms are equal.

However, the longest proper arithmetic progression in the spectrum has length only \( O(n) \). Thus, despite the presence of a long non-proper progression due to the existence of an eigenvalue with high multiplicity, the structure relevant to \cref{thm:general comparison theorem dependent only on length of arithmetic sequence}, the proper arithmetic progressions, remains short.
\end{remark}

\subsection{Related Work and Significance}
\


Besides \cite{temme13}, several works considered the question of comparing quantum and classical spectral gaps \cite{chen2021fast,ramkumar2024mixing,chen2024randomized}, albeit under strong assumptions on the Hamiltonians and/or the Lindbladian dynamics. In particular, \cite{chen2021fast,ramkumar2024mixing,chen2024randomized} study Lindbladian dynamics defined by random jumps. While \cite{chen2024randomized} studies the Davies generator, \cite{chen2021fast,ramkumar2024mixing} study a physically-motivated variant of the Davies generator known as the CKG generator \cite{chen2023efficient}. The analysis in \cite{chen2021fast,ramkumar2024mixing} relies on bounding the spectral gap for the expected CKG generator over all random jumps; under their assumptions, the expected CKG generator coincides with the expected Davies generator up to a rescaling of the transition rate.
For any eigenbasis $U$ of $H,$ the expected Davies generator studied in \cite{chen2021fast,ramkumar2024mixing,chen2024randomized} preserves the subspace of observables that are diagonal in $U$, while factoring the non-diagonal observables into invariant subspaces of dimension $1$ i.e. $ \mathcal{L}(\ket{u_i}\bra{u_j})\subseteq \text{span}(\ket{u_i}\bra{u_j})$ for any distinct eigenvectors $u_i,u_j$ in $U$.\footnote{\cite{rajakumar2024gibbs} presents several results, of which the above assertion applies to all except the cyclic graph and hypercube Hamiltonians. In the cyclic graph case, the expected CKG generator coincides with the expected Davies generator, allowing our comparison result for Davies generators to streamline their analysis. The hypercube case, however, relies on a different approach tailored to their specific choice of jump operators.} In each case, these non-diagonal invariant subspaces are then analyzed directly via adaptations of Temme's argument and shown to be comparable to the associated classical spectral gap. In hindsight this comparison follows more conceptually from the following straightforward extension of Temme's result (specifically \cite[Theorem 11]{temme13} 
combined with a standard bottleneck argument applied to the classical Markov generator). Although the argument is elementary, it appears to have gone unnoticed in the literature. We include a proof in \cref{sec:missing proofs}.

\begin{proposition}[Extension of \cite{temme13}]\label{prop:simple comparison result when orbit has dimension 1}
    Consider a Hamiltonian $H,$ orthonormal eigenbasis $U =\{u_i\}$ of $H $ and a Davies generator $\mathcal{L}$ associated with $H.$ Suppose that for all
    $i\neq j$, $\mathcal{L}(\ket{u_i}\bra{u_j})\subseteq \text{span} (\ket{u_i}\bra{u_j})$. Then, 
    $ \lambda_{\mathcal{L}}\geq \frac{1}{2}\lambda_{\mathcal{L},U, \mathrm{cl}}.$
    In particular, when $ H$ has  simple spectrum and simple Bohr spectrum, $\lambda_{\mathcal{L}} \geq \frac{1}{2}\lambda_{\mathcal{L},\mathrm{cl}}$.
\end{proposition}

In this context, the algorithmic significance of our work is threefold. First, one can imagine scenarios where the expected Lindbladian in the above scheme is a Davies generator corresponding to a Hamiltonian with degenerate Bohr spectrum, but satisfying the short AP property. Second, Davies generators of commuting Hamiltonians are local and efficiently implementable, and can be directly used as algorithms; the short AP property may be easy to verify and useful in their analysis. Finally, the ability to restrict attention to observables in $V_0$ may be useful in developing further proof techniques to understand the mixing of Davies and related generators.

In the context of physics, there is a physical intuition which suggests that in this context ``dephasing'' (corresponding to off-diagonal decay)  should happen at time scales shorter than mixing \cite{zurek2002decoherence,zurek2003decoherence,Schlosshauer2007Decoherence}. 
 \cref{thm:general comparison theorem dependent only on length of arithmetic sequence} shows that this is true up to a small factor for Davies generators of short AP Hamiltonians, clarifying the mechanism of thermalization for such generators. Equivalently, the failure of rapid dephasing can always be blamed on a long arithmetic progression in the spectrum of the Hamiltonian. 
 
While it is easy to construct artificial local Hamiltonians with long arithmetic progressions (e.g. $H=\sum_{i\le n} 2^{-i}Z_i$), we are not aware of any natural physical $n$-qubit local Hamiltonians whose spectra contain arithmetic progressions of length longer than $O(n)$, with this bound being met by the Toric code \cite{kitaev2003fault}\footnote{Since the toric code Hamiltonian consists of pairwise commuting terms where each has eigenvalues $\pm 1,$ it is easy to explicitly compute the Hamiltonian spectrum.}. A natural research direction is to show that some large class of physical Hamiltonians --- such as the transverse field Ising model on arbitrary graphs --- satisfy the short AP property with $D=\mathrm{poly}(n)$, which would yield a useful quantum to classical comparison for this class.


\subsection{Techniques}

\subsubsection{Proof of main comparison theorem  ( \texorpdfstring{\cref{thm:compare spectral gap V omega and V 0,thm:general comparison theorem dependent only on length of arithmetic sequence}})}

We outline the strategy for proving \cref{thm:compare spectral gap V omega and V 0}. Given an observable \( f \in V_\omega \) where $\omega \neq 0$, we construct two observables:
\[
g = e^{\frac{\beta\omega}{4}} (f f^\dag)^{1/2}, \quad \text{and} \quad h = e^{-\frac{\beta\omega}{4}} (f^\dag f)^{1/2}.
\]
Both \( g \) and \( h \) are Hermitian and belong to \( V_0 \). Our strategy is to (1) lower bound the Dirichlet form of \( f \) by the Dirichlet forms of \( g \) and \( h \), then (2) relate the Dirichlet forms of $g$ and $h$ to their variances, and (3) lower bound the variance of $g$ and $h$ by the variance of $f.$ For (1), we rewrite the Dirichlet form of $f,g,h$ as weighted commutator norms $\| \rho^{1/4}[S(\tilde{\omega}),f]\rho^{1/4}\|_F^2$ then apply a novel trace inequality whose proof utilizes the spectral decomposition of $f, g, h.$ (2) is simply due to $g,h$ being observables in $V_0$ and the definition of $\lambda_{\mathcal{L},0}.$  For (3), we again use a spectral decomposition of $f,g,h,$ and further note that the vectors appearing in this decomposition are eigenvectors of the Hamiltonian $H,$ and use the relationship between their eigenvalues to lower bound the variances of $g, h.$ 

Specifically,
\( g = e^{\frac{\beta\omega}{4}} \sum \lambda_i \ket{u_i}\bra{u_i} \), \( 
f = \sum \lambda_i \ket{u_i}\bra{v_i} \), and \( h = e^{-\frac{\beta\omega}{4}} \sum \lambda_i \ket{v_i}\bra{v_i} \), 
with \( u_i, v_i \) being eigenvectors of both \( H \) and \( \rho \) and the corresponding eigenvalues of $H$ satisfy \( H(v_i) = H(u_i) - \omega\). 

Intuitively, lower-bounding the variances of $g$ and $h$ is equivalent to showing that $g$ and $h$ are far from (multiple of) the identity. More concretely, using the above decomposition and Hölder's inequality, and let $\rho(v_i) ,\rho(u_i)$ be the corresponding eigenvalues of $\rho$, we get:
\[
\Var(g) \geq (1 - \sum \rho(u_i)) \Var(f), \quad 
\Var(h) \geq (1 - \sum \rho(v_i)) \Var(f)
\]
and since \(1\geq \sum \rho(v_i) = e^{\beta \omega} \sum \rho(u_i) \) and \(1\geq \sum \rho(u_i) = e^{-\beta \omega} \sum \rho(v_i) \), we conclude:
\[
\max(\Var(g), \Var(h)) \geq (1 - e^{-|\beta \omega|}) \Var(f).
\]

Next, we provide some intuition on why $\max(\Var(g), \Var(h)) \geq \frac{1}{D} \Var(f)$ where $D$ is the length of the longest arithmetic progression with common difference $\omega$ in the spectrum of $H.$
Take a sequence of eigenvectors $u_{i_1}, \dots, u_{i_k}$, where $H(u_{i_{j+1}}) = H(u_{i_j}) -\omega$ for all  $j$; then $H(v_{i_k}) = H(u_{i_k}) -\omega$ is also in $\spec(H).$ By assuming that this chain is maximal, we conclude that $v_{i_k}$ is not in the spectral decomposition of $g.$ 
We can assume w.l.o.g.\ that $\beta \omega \geq 0,$ in which case we can argue that $\rho(v_{i_k}) \geq \frac{1}{k}\sum_{j=1}^k \rho(u_{i_j}).$ Additionally, note that the  $k+1 \leq D$ since $H(u_{i_1}), \cdots, H(u_{i_k}), H(v_{i_k})$ forms a $(k+1)$-arithmetic progression with common difference $\omega$ in $\spec(H)$; these observations are enough to conclude that $1 -\sum \rho(u_i)\geq 1/D$ and thus $\Var(g) \geq \frac{1}{D}\Var(f).$
\subsubsection{Proof that generic Hamiltonians do not have long arithmetic progressions}
Finally, we sketch the proof of \cref{thm:generic transverse field}, following the approach in \cite{HuangHarrow25}.  Let \( \lambda_1, \dots, \lambda_N \) be the eigenvalues of a Hamiltonian \( H \). Observe that \( H \) has no 3-term arithmetic progression in its spectrum if and only if
\[
F := \prod_{\substack{i, j, k \text{ distinct}}} (\lambda_i + \lambda_j - 2\lambda_k)
\]
is nonzero.

Since \( F \) is symmetric in the eigenvalues, it can be viewed as a polynomial in the entries of \( H \). Thus, the set of Hamiltonians for which \( F = 0 \) has measure zero, unless \( F \equiv 0 \) identically. To rule this out, it suffices to exhibit a single instance where \( F \ne 0 \); such an example can be constructed for each type of perturbation we consider.

We remark that this approach generalizes to the case of longer \( D \)-term arithmetic progressions. In that case, the role of the factor \( (\lambda_i + \lambda_j - 2\lambda_k) \) is replaced by a suitable expression such as
\[
\sum_{j=1}^{D-2} \left( (\lambda_{i_j} - \lambda_{i_{j+1}}) - (\lambda_{i_{j+1}} - \lambda_{i_{j+2}}) \right)^2,
\]
which similarly vanishes when a \( D \)-term arithmetic progression is present. We can also prove that $H$ does not have repeated eigenvalues using a similar strategy, where we replace $F$ with $G = \prod_{i\neq j}(\lambda_i -\lambda_j) $.

\subsection{Organization}
In \cref{sec:general spectral comparison}, we prove \cref{thm:compare spectral gap V omega and V 0} and \cref{thm:general comparison theorem dependent only on length of arithmetic sequence}. In \cref{sec:quantum to classical spectral gap}, we prove \cref{prop:classical chain V0,thm:quantum vs classical spectral gap comparison,cor:cheeger}. In \cref{sec:no AP in perturbed Hamiltonian}, we prove \cref{thm:generic transverse field} and related results. 

%% file: prelim.tex
\section{Preliminaries}
\subsection{Notation}
Let \( n \) be the number of qubits, and let \( N = 2^n \) denote the dimension of the Hilbert space. We write \( \rho \) for the quantum Gibbs state, use \( \langle \cdot, \cdot \rangle_\rho \) to denote the KMS inner product, and define the associated norm by
\[
\|A\|_\rho^2 := \langle A, A \rangle_\rho = \|\rho^{1/4} A \rho^{1/4}\|_{\mathrm{F}}^2,
\]
where \( \|\cdot\|_{\mathrm{F}} \) denotes the Frobenius norm.

For \(k \in \mathbb{N}_{\geq 1}\), we write \([k]\) for the index set \(\{1, \dots, k\}\).
 We use \(A^\dagger\) to denote the conjugate transpose of a matrix \(A\).
We will also occasionally use \emph{braket notation}, where \(\ket{u}\) denotes a column vector and \(\bra{u}\) its conjugate transpose.

\subsection{Linear algebra}

\begin{proposition}\label{prop:eigenvalue multiplicities of AB and BA}
Let $A, B$ be arbitrary matrices. Then the nonzero spectrum of $AB$ and $BA$ are the same, i.e., $\mathrm{spec}(AB) \cap \R_{\neq 0}= \mathrm{spec}(BA)\cap \R_{\neq 0}.$ In other words, if $ \lambda\neq 0$ is an eigenvalue of $AB$ with multiplicity $m$, then it is also an eigenvalue of $BA$ with multiplicity $m,$ thus
 $\trace{AB} =\trace{BA}.$ 
\end{proposition}

\begin{proof}
It's enough to check that, if $v_1, \dots, v_m$ are linearly independent eigenvectors of $AB$ corresponding to an eigenvalue $\lambda \neq 0,$ then $Bv_1, \dots, Bv_m$ are linearly independent eigenvectors of $BA$ corresponding to the eigenvalue  $\lambda.$ First, observe that, for each $v_i,$ $ BA (B v_i) = B (AB v_i) = \lambda (Bv_i).$ Suppose $\set{Bv_i}_i$ are not linearly independent; then there exist nontrivial coefficients $\set{c_i}_i$ such that $ \sum_i c_i (B v_i)=0,$ but this means $ \lambda \sum_{i} c_i v_i = AB \sum c_i v_i = A (B\sum c_i v_i) = 0,$ thus $ \sum c_i v_i =0,$ contradicting the assumption that $\set{v_i}_{i=1}^m$ are linearly independent.
\end{proof}

\subsection{Dirichlet form and spectral gap of the Davies generator}
We begin by reviewing some basic properties of the subspace of observables \( V_\omega \) (as defined in \cref{def:V omega}), and then establish the relationship between the Lindbladian spectral gap \( \lambda_{\mathcal{L}} \) and the spectral gap \( \lambda_{\mathcal{L},\omega} \) in \( V_\omega \), as stated in \cref{eq:spectral gap is minimize over all subspace}. These properties have been observed in \cite{temme13}, but we include the proofs for the sake of completeness.
\begin{proposition} \label{prop:product rule}
    If $f\in V_\omega$ and $g\in V_\omega'$ then $ f g \in V_{\omega + \omega'}.$
\end{proposition}
\begin{proof}
Note that $ f = \sum_{\lambda} \Pi_{\lambda+\omega} f \Pi_{\lambda}$ and $g = \sum_{\lambda'} \Pi_{\lambda'+\omega'} g \Pi_{\lambda'} $ and $\Pi_{\lambda} \Pi_{\lambda'} = \Pi_{\lambda} \delta_{\lambda = \lambda'} $, then
\begin{align*}
    fg = ( \sum_{\lambda} \Pi_{\lambda+\omega} f \Pi_{\lambda}) (\sum_{\lambda'} \Pi_{\lambda'+\omega'} g \Pi_{\lambda'} ) = \sum_{\lambda} \Pi_{\lambda+ \omega} f \Pi_{\lambda} g \Pi_{\lambda -\omega'} \in V_{\omega+\omega'}.\tag*{\qedhere}
\end{align*}
\end{proof}
\begin{proposition}\label{prop:orthogonality of different subspace}
If $ h\in V_\omega$ with $\omega \neq 0$ then $\trace{h} =0.$
    Consequently, $f\in V_\omega$ and $g\in V_{\omega'}$ with $\omega -\omega' \neq 0$ then $ \langle f, g\rangle_\rho = 0.$
\end{proposition}
\begin{proof}
For any $ h \in V_{\tilde{\omega}}$ with $\tilde{\omega} \neq 0,$ we can write
    $ h = \sum_{\lambda} \Pi_{\lambda+\tilde{\omega}} h \Pi_{\lambda} $ and since $\Pi_{\lambda} \Pi_{\lambda+ \tilde{\omega}}  = 0,$
    \[\trace{h} = \sum_{\lambda} \trace{\Pi_{\lambda+\tilde{\omega}} h \Pi_{\lambda}} =\sum_{\lambda} \trace{  \Pi_\lambda \Pi_{\lambda+\tilde{\omega}}h} = 0.\]
    By \cref{prop:product rule}, $f \rho^{1/2} g^\dag \rho^{1/2} \in V_{\omega -\omega'},$ and $\omega -\omega' \neq 0$ so $\langle f, g\rangle_\rho = \trace{f \rho^{1/2} g^\dag \rho^{1/2}}= 0.$ 
\end{proof}
\begin{proposition}\label{prop:V omega is invariant subspace of Davies generator}
Consider a Hamiltonian $H$ and a Davies generator $\mathcal{L}$ associated with $H.$ Then, $V_\omega$  is an invariant subspace of $\mathcal{L}$; i.e., $ \mathcal{L}(V_\omega)\subseteq V_\omega.$
\end{proposition}
\begin{proof}
    Fix $f \in V_\omega.$ We only need to show that for any Bohr frequency $\omega',$ $\mathcal{L}_{\omega',S}(f) \in V_{\omega},$ which is true by \cref{prop:product rule} and the fact that $S(\omega')\in V_{\omega'}, S(\omega')^\dag \in V_{-\omega'}.$
\end{proof}
The following relate the spectral gap of the Lindbladian to the spectral gap on the invariant subspaces $ V_\omega$.
\begin{proposition}\label{prop:relate spectral gap of Lindbladian to spectral gap in subspace Vomega}
    Let $\lambda_{\mathcal{L},\omega}$ be as defined in \cref{eq:spectral gap on subspace}. Then,
    \[\lambda_\mathcal{L} = \min_\omega \lambda_{\mathcal{L},\omega}.\]
\end{proposition}
To prove \cref{prop:relate spectral gap of Lindbladian to spectral gap in subspace Vomega}, we need the following proposition.
\begin{proposition}\label{prop:decomposition of dirichlet form and variance}
For any arbitrary observable $f\in \C^{N\times N},$ we have
\begin{equation}
    f = \sum_\omega f(\omega)\nonumber
\end{equation} where the sum is over Bohr frequencies $\omega$ and $f(\omega)$ is as defined in \cref{def:V omega}. Moreover,
 \begin{equation}
     \mathcal{E}_{\mathcal{L}}(f) = \sum_\omega \mathcal{E}_{\mathcal{L}}(f(\omega)) \nonumber
 \quad
 and 
 \quad
     \Var(f) = \sum_\omega \Var(f(\omega)). \nonumber
 \end{equation}    
\end{proposition}
\begin{proof}
For the first statement, since $ I = \sum_{\lambda} \Pi_\lambda,$ we can write
\[ f=(\sum_{\lambda_1} \Pi_{\lambda_1}) f (\sum_{\lambda_2} \Pi_{\lambda_2}) = \sum_{\lambda_1, \lambda_2} \Pi_{\lambda_1} f \Pi_{\lambda_2} = \sum_{\omega} \sum_{\lambda_1, \lambda_2:\lambda_1 -\lambda_2 =\omega} \Pi_{\lambda_1} f \Pi_{\lambda_2} =\sum_{\omega}  f(\omega)\]
    By \cref{prop:V omega is invariant subspace of Davies generator,prop:orthogonality of different subspace},
\[ \mathcal{E}_{\mathcal{L}}(f) = -\left\langle \sum_\omega f(\omega),  \sum_\omega\mathcal{L}(f(\omega))\right\rangle_\rho =-\sum_\omega \langle f(\omega), \mathcal{L}(f(\omega))\rangle_\rho =\sum_\omega \mathcal{E}_{\mathcal{L}} (f(\omega)).\]
Note that by \cref{prop:orthogonality of different subspace}, $ \trace{\rho f} = \trace{\rho f(0)}$ since $\trace{\rho f(\omega)} = 0\forall \omega \neq 0,$ and $\langle f(\omega), f(\omega')\rangle_\rho =0$ if $\omega \neq \omega'.$ We thus obtain: 
\begin{align*}
    \Var(f) &= \langle f, f\rangle_\rho - \trace{\rho f} \trace{\rho f^\dag} =   \langle f, f\rangle_\rho - \trace{\rho f{(0)}} \trace{\rho (f{(0)})^\dag} \\
    &= \sum_\omega \langle f(\omega) ,f(\omega) \rangle_\rho - \trace{\rho f{(0)}} \trace{\rho (f{(0)})^\dag} \\
    &= \sum_\omega \Var( f(\omega) ). \tag*{\qedhere}
\end{align*}
\end{proof}

\begin{proof}[Proof of \cref{prop:relate spectral gap of Lindbladian to spectral gap in subspace Vomega}]
Obviously, for all $\omega$, \[\lambda_{\mathcal{L},\omega}=\min_{f\in V_\omega} \lambda(f) \geq \min_f \lambda(f) = \lambda_{\mathcal{L}},\quad
\text{thus} \quad \min_\omega \lambda_{\mathcal{L},\omega} \geq \lambda_{\mathcal{L}}.\]
Let $f$ be such that $\lambda(f) =\frac{\mathcal{E}_{\mathcal{L}}(f)}{\Var(f) }=\lambda_\mathcal{L} . $ Using \cref{prop:decomposition of dirichlet form and variance}, write $f =\sum_\omega f(\omega)$ with $f(\omega) \in V_\omega.$ Then 
\begin{align*}
    \lambda_{\mathcal{L}} = \frac{\sum_\omega \mathcal{E}_{\mathcal{L}}(f(\omega))}{\sum_\omega \Var(f(\omega)) } \geq \min_\omega \lambda_{\mathcal{L}}(f(\omega)) \geq \min_\omega \lambda_{\mathcal{L},\omega}. \tag*{\qedhere}
\end{align*}
\end{proof}

\begin{proposition} \label{prop:operator in V omega commutation with rho}
    Recalling that  $\rho \equiv \rho_\beta =\frac{\exp(-\beta H)}{\tr(\exp(-\beta H))}$. For any observable $f\in V_\omega$ and $m\in \R,$
    \[ \rho^m f = f \rho^m \exp(-\beta m \omega).\]
\end{proposition}
\begin{proof}
    For $f\in V_\omega$, we have $f = f(\omega) = \sum_{\lambda} \Pi_{\lambda+\omega} f \Pi_{\lambda},$ so that 
    \begin{align*}
        e^{-\beta m H} f &= \sum_{\lambda}  e^{-\beta m H}   \Pi_{\lambda+\omega} f \Pi_{\lambda} = \sum_{\lambda} e^{-\beta m (\lambda+ \omega) }\Pi_{\lambda+\omega} f \Pi_{\lambda} \\
        &= e^{-\beta m  \omega } \sum_{\lambda} \Pi_{\lambda+\omega} f \Pi_{\lambda}  e^{-\beta  m\lambda} =e^{-\beta m  \omega } \sum_{\lambda} \Pi_{\lambda+\omega} f \Pi_{\lambda} e^{-\beta m H} = e^{-\beta m  \omega } f   e^{-\beta m H};
    \end{align*}
    thus 
    \[ \rho^m f = e^{-\beta m  \omega } f \rho^m. \qedhere \]
\end{proof}
\begin{proposition} \label{prop:observable in V0 commute with $H$}
    $f\in V_0$ if and only if $f$ commutes with $H.$
\end{proposition}
\begin{proof}
    Consider $f \in V_0,$ then $f = \sum_{\lambda} \Pi_{\lambda} f \Pi_{\lambda}$ and $H \Pi_{\lambda} = \lambda\Pi_{\lambda} = \Pi_{\lambda} H$ thus
    \[ f H = \sum_{\lambda} \Pi_{\lambda} f \Pi_{\lambda} H  = \sum_{\lambda} \lambda  \Pi_{\lambda} f \Pi_{\lambda} =   \sum_{\lambda} H \Pi_{\lambda} f \Pi_{\lambda}  = H f.  \]
    Consider $f$ that commutes with $H.$ Easy to see that $f^\dag$ also commutes with $H,$ since $ f^\dag H = (H f)^\dag = (fH)^\dag = H f^\dag.$ Thus, $\frac{f+f^\dag}{2}$ and $\frac{i(f-f^\dag)}{2}$ also commute with $H.$ Since $\frac{f+f^\dag}{2}$ is Hermitian and commutes with $H,$ $\frac{f+f^\dag}{2}$ and $H$ are simultaneously diagonalizable by some eigenbasis $\{u_j\}$ of $H,$ so 
    \[ \frac{f+f^\dag}{2} = \sum_{j}\gamma_j \ket{u_j}\bra{u_j}. \]
    Similarly, for some eigenbasis $\{u'_i\}$ of $H$,
    \[\frac{i(f-f^\dag)}{2} = \sum_{j}\gamma'_j \ket{u'_j}\bra{u'_j},\]
    thus $ f = \sum_j \gamma_j \ket{u_j}\bra{u_j} - i\sum_j \gamma'_j \ket{u'_j}\bra{u'_j}$ so $f \in V_0.$
\end{proof}

We note that the Dirichlet form can be written as a sum of norms of commutators. Previously, \cite[Proposition 2.5]{Carlen2018NoncommutativeCO} showed a similar result for the special case $G(\omega) = \exp(-\frac{\beta\omega}{2}).$ Related identities for other specific choices of the transition rate function $G$ can be found in \cite{rouze2025efficient, chen2025quantum}. For completeness, we include a proof in \cref{sec:missing proofs}.
\begin{proposition}[Dirichlet in divergence form]\label{prop:dirichlet form as sum of commutators}


    For any Hamiltonian $H,$ and Davies generator $\mathcal{L}$ as defined in \cref{def:davies dissipative part},
    the Dirichlet form $\mathcal{E}_{\mathcal{L}}(f) =-\langle  \mathcal{L}(f),f\rangle_\rho$ 
   can be rewritten as \[\mathcal{E}_{\mathcal{L}}(f)= \frac{1}{2}\sum_{S\in \mathcal{S}, \omega} \tilde{G}(\omega)\|[S(\omega), f]\|_\rho^2, \] where $\tilde{G}(\omega) = G(\omega) e^{\frac{\beta \omega }{2}}.$ 
\end{proposition}


We can now show that the spectral gap is unchanged if restricted to Hermitian observables.
\begin{proposition}\label{prop:compare spectral gap with hermitian spectral gap in V0}

    Let $\lambda_{\mathcal{L}, 0}, \lambda_{\mathcal{L}, 0}^{\mathrm{H}}$ be as defined in \cref{eq:spectral gap on subspace,eq:spectral gap hermitian} respectively. Then
    \[\lambda_{\mathcal{L}, 0} =\lambda_{\mathcal{L}, 0}^{\mathrm{H}}. \]
\end{proposition}
We need the following proposition.
\begin{proposition}\label{prop:compare spectral gap with hermitian spectral gap helper}
With the same setup as \cref{prop:dirichlet form as sum of commutators}, 
\[\mathcal{E}_{\mathcal{L}}(f) = \mathcal{E}_{\mathcal{L}}(f^\dag) \quad \text{and} \quad \Var(f) = \Var(f^\dag). \]
\end{proposition}
\begin{proof}
    Since $S^\dag(-\omega)  = S(\omega)^\dag, $ we have 
    \[ [S^\dag(-\omega) , f^\dag] = -[S(\omega), f]^\dag \] and thus for $A: =\rho^{1/4} [S(\omega), f] \rho^{1/4}$,
    \[ \|[S(\omega),f]\|_\rho^2 = \|A\|_{\mathrm{F}}^2 = \|A^\dag \|_{\mathrm{F}}^2 = \|[S^\dag(-\omega) , f^\dag]\|_\rho^2.\]
    Combining this with $\tilde{G}(\omega ) = \tilde{G}(-\omega),$ we have
    \[\mathcal{E}_{\mathcal{L}}(f) = \frac{1}{2}\sum_{S\in \mathcal{S},\omega} \tilde{G}(\omega )\|[S(\omega),f]\|_\rho^2 = \frac{1}{2}\sum_{S\in \mathcal{S},\omega} \tilde{G}(-\omega )\|[S^\dag(-\omega),f^\dag]\|_\rho^2 = \mathcal{E}_{\mathcal{L}}(f^\dag). \]
    Next, 
    \[ \Var(f) = \trace{f \rho^{1/2} f^\dag \rho^{1/2} } - \trace{\rho f}\trace{\rho f^\dag} = \trace{f^\dag \rho^{1/2} f\rho^{1/2} } -\trace{\rho f^\dag} \trace{\rho f} = \Var(f^\dag).\qedhere \]
\end{proof}
\begin{proof}[Proof of \cref{prop:compare spectral gap with hermitian spectral gap in V0}]
Obviously $ \lambda_{\mathcal{L}, 0}^{\mathrm{H}}\geq\lambda_{\mathcal{L}, 0}. $
We prove $ \lambda_{\mathcal{L}, 0}\geq \lambda_{\mathcal{L}, 0}^{\mathrm{H}}.$ Fix $f \in V_0$ such that $\frac{\mathcal{E}_{\mathcal{L}}(f)}{\Var(f)} = \lambda_{\mathcal{L},0}.$ 
We show that 
\[\mathcal{E}_{\mathcal{L}}(f) = \mathcal{E}_{\mathcal{L}}\left( \frac{f + f^\dag}{2}\right)+ \mathcal{E}_{\mathcal{L}}\left(\frac{i(f-f^\dag)}{2}\right).\]
Indeed, by \cref{prop:dirichlet form as sum of commutators,prop:compare spectral gap with hermitian spectral gap helper}, the LHS can be rewritten as
\begin{align*}
  \mathcal{E}_{\mathcal{L}}(f) =  \frac{\mathcal{E}_{\mathcal{L}}(f) + \mathcal{E}_{\mathcal{L}}(f^\dag)}{2} &= \frac{1}{2}\sum_{S\in \mathcal{S},\omega} \tilde{G}(\omega) \frac{\|[S(\omega), f]\|_\rho^2 + \|[S(\omega), f^\dag]\|_\rho^2}{2}\\
    &=_{(1)}\frac{1}{2} \sum_{S\in \mathcal{S},\omega} \tilde{G}(\omega)  \left( \left\|\left[S(\omega), \frac{f + f^\dag}{2}\right]\right\|_\rho^2 + \left\|\left[S(\omega), \frac{i(f-f^\dag)}{2}\right]\right\|_\rho^2\right)\\
    &= \mathcal{E}_{\mathcal{L}}\left(\frac{f + f^\dag}{2}\right)+ \mathcal{E}_{\mathcal{L}}\left(\frac{i(f - f^\dag)}{2}\right)
\end{align*}
where (1) is due to the identity $ \trace{A A^\dag } + \trace{B B^\dag} = \frac{1}{2}(\trace{(A+B) (A+B)^\dag}  + \trace{(A-B)(A-B)^\dag}), $ where $A = \rho^{1/4} [S(\omega) , f]\rho^{1/4}$ and $ B = \rho^{1/4} [S(\omega) , f^\dag]\rho^{1/4}.$

Next, we show
\begin{equation}\label{eq:compare variance of nonhermitian and hermitian}
    \Var\left(\frac{f+f^\dag}{2}\right) + \Var\left(\frac{i(f-f^\dag)}{2}\right) = \Var(f) \nonumber
\end{equation}
and using the fact that $ \frac{f+f^\dag}{2}, \frac{i(f-f^\dag)}{2} $  are Hermitian and in $V_0$, conclude that 
\[ \mathcal{E}_{\mathcal{L}}(f) = \mathcal{E}_{\mathcal{L}}\left( \frac{f + f^\dag}{2}\right)+ \mathcal{E}_{\mathcal{L}}\left(\frac{i(f-f^\dag)}{2}\right) \geq  \lambda_0^{\mathrm{H}} \left(\Var \left( \frac{f + f^\dag}{2}\right) + \Var\left(\frac{i(f-f^\dag)}{2}\right)\right) = \lambda_{\mathcal{L},0}^{\mathrm{H}} \Var(f),\]
which implies $\lambda_{\mathcal{L},0} \geq \lambda_{\mathcal{L},0}^{\mathrm{H}}$ due to the choice of $f.$

Indeed, 
\begin{align*}
     &\Var\left(\frac{f+f^\dag}{2}\right) + \Var\left(\frac{i(f-f^\dag)}{2}\right) \\
     &\qquad = \frac{1}{4}\left(\langle f + f^\dag, f+f^\dag\rangle_\rho - \trace{\rho (f+ f^\dag) }^2 +  \langle f - f^\dag, f-f^\dag \rangle_\rho - \trace{\rho (f- f^\dag) } \trace{\rho(f^\dag- f) }\right)\\
     &\qquad = \frac{1}{2} (\langle f , f\rangle\rho + \langle f^\dag , f^\dag\rangle\rho) - \trace{ \rho f } \trace{\rho f^\dag} \\
     &\qquad = \Var(f). \tag*{\qedhere}
\end{align*}
\end{proof}
With a similar proof,\footnote{Take arbitrary $f$ such that $ \lambda(f) = \lambda_{\mathcal{L}}$ instead of $ f\in V_0$ such that $\lambda(f) = \lambda_{\mathcal{L},0}$.} we obtain the following. 
\begin{proposition}
    \label{prop:compare spectral gap with hermitian spectral gap}
    Let $\lambda_{\mathcal{L}}$, $ \lambda_{\mathcal{L}}^{\mathrm{H}}$ be as defined in \cref{eq:spectral gap,eq:spectral gap hermitian} respectively. Then
    \[\lambda_{\mathcal{L}} =\lambda_{\mathcal{L}}^{\mathrm{H}}. \]
\end{proposition}
\begin{remark}
    An alternative proof of \cref{prop:compare spectral gap with hermitian spectral gap,prop:compare spectral gap with hermitian spectral gap in V0} is by viewing $\lambda_{\mathcal{L}}$ as the second smallest eigenvalue of $ \mathcal{L},$ and $ \lambda_{\mathcal{L}, 0} $ as the second smallest eigenvalue of the linear operator obtained by restricted $\mathcal{L}$ to the invariant subspace $V_0.$ If $f$ is an eigenvector of $\mathcal{L}$ with eigenvalue $\lambda_{\mathcal{L}} ,$ then so is $ f^\dag$ by  \cref{prop:compare spectral gap with hermitian spectral gap helper} and the minimax principle. The linear combination $\frac{f + f^\dag}{2} $ is Hermitian, and is also an eigenvector $\mathcal{L}$ with eigenvalue $\lambda_{\mathcal{L}},$ yielding that $\lambda_{\mathcal{L}} = \lambda_{\mathcal{L}}^{\mathrm{H}}. $ Similarly, if $f\in V_0$ is an eigenvector of $\mathcal{L}$ with eigenvalue $\lambda_{\mathcal{L},0},$ then the Hermitian observable $\frac{f + f^\dag}{2} \in V_0$ is also an an eigenvector $\mathcal{L}$ with eigenvalue $\lambda_{\mathcal{L},0},$ implying that $\lambda_{\mathcal{L},0} = \lambda_{\mathcal{L},0}^{\mathrm{H}}. $
\end{remark}

Some formulations of the Davies generator include the coherent Hamiltonian evolution explicitly, writing the generator as  
\[
\tilde{\mathcal{L}}(f) = i[H, f] + \mathcal{L}(f),
\]
where \( \mathcal{L} \) is as \cref{def:davies generator}. The following proposition clarifies that the inclusion of the coherent term \( i[H, f] \) has no effect on the Dirichlet form, and thus does not affect the spectral gap.

\begin{proposition}\label{prop:incoherent term is not important}
Let \( \tilde{\mathcal{L}}(f) = i[H, f] + \mathcal{L}(f) \), where \( \mathcal{L} \) is the Davies generator as defined in \cref{def:davies generator}, and let \( \rho \) be the stationary state. Then:
\begin{enumerate}
    \item For all observables \( f \), the real part of the Dirichlet form is unaffected by the coherent term:
    \[
    \operatorname{Re} ( \langle f, \tilde{\mathcal{L}}(f) \rangle_\rho) = \langle f, \mathcal{L}(f) \rangle_\rho.
    \]
    \item In particular, for any Hermitian observable \( f = f^\dag \), the Dirichlet forms of \( \tilde{\mathcal{L}} \) and \( \mathcal{L} \) coincide:
    \[
    \langle f, \tilde{\mathcal{L}}(f) \rangle_\rho = \langle f, \mathcal{L}(f) \rangle_\rho.
    \]
\end{enumerate}
\end{proposition}

\begin{proof}
    We only need to check that 
    \[ \langle [H,f], f\rangle_\rho \in \R \quad \text{and} \quad \text{for all $f$ such that $f=f^\dag$, } \langle [H,f], f\rangle_\rho = 0. \]
    Let $\tilde{f} = \rho^{1/4} f \rho^{1/4}.$
    Since $\rho $ commutes with $H,$ we can rewrite
    \begin{align*}
        \langle [H,f], f\rangle_\rho =\trace{\rho^{1/4}H f \rho^{1/2} f^\dag \rho^{1/4}} -\trace{\rho^{1/4}f H \rho^{1/2} f^\dag \rho^{1/4}} = \trace{\tilde{f}^\dag H \tilde{f}  }-\trace{ \tilde{f}H\tilde{f}^\dag } = \trace{\tilde{f}^\dag H \tilde{f} - \tilde{f}H\tilde{f}^\dag} \in\R
    \end{align*}
    since $\tilde{f}^\dag H \tilde{f} - \tilde{f}H\tilde{f}^\dag$ is Hermitian, thus has real trace. If $f = f^\dag$ then $\tilde{f}=\tilde{f}^\dag,$ so the above yields \[ \langle [H,f], f\rangle_\rho = 0.\qedhere\]
\end{proof}

%% file: general_comparison_theorem.tex
\section{Proof of main comparison theorems} \label{sec:general spectral comparison}

In this section, we prove \cref{thm:general comparison theorem dependent only on length of arithmetic sequence,thm:compare spectral gap V omega and V 0}.

For a Bohr frequency $\omega$, and an operator $f\in V_\omega,$ we let 
\begin{equation}\label{eq:hermitianization of f}
    g := e^{\frac{\beta\omega}{4}} (f f^\dag)^{1/2} \quad \text{and} \quad h := e^{-\frac{\beta\omega}{4}} (f^\dag f)^{1/2}
\end{equation}
\begin{theorem}[Comparing Dirichlet forms] \label{thm:comparison of dirichlet forms}
    For any $f\in V_\omega,$ and $g,h$ as defined in \cref{eq:hermitianization of f}.
      For any Davies generator $\mathcal{L},$
\[2 \mathcal{E}_\mathcal{L}(f)\geq \mathcal{E}_\mathcal{L}(g) + \mathcal{E}_\mathcal{L}(h).\]

\end{theorem}
To prove \cref{thm:comparison of dirichlet forms}, we need the following propositions.
\begin{proposition} \label{prop:commutator decomposition}
Fix an arbitrary observable $ \phi.$
   Let $\tilde{\phi} := \rho^{1/2} \phi \rho^{1/2}$. Then for any Bohr frequency $\tilde{\omega}$,
    \[\|[S(\tilde{\omega}), \phi]\|_\rho^2 = e^{-\frac{\beta \tilde{\omega}}{2}} 
 \trace{S(\tilde{\omega}) \tilde{\phi} \tilde{\phi}^\dag S(\tilde{\omega})^\dag } + e^{\frac{\beta \tilde{\omega}}{2}} 
 \text{tr}(S(\tilde{\omega})^\dag \tilde{\phi}^\dag \tilde{\phi} S(\tilde{\omega})) -  \trace{S(\tilde{\omega}) \tilde{\phi} S(\tilde{\omega})^\dag \tilde{\phi}^\dag} -  \trace{S(\tilde{\omega}) \tilde{\phi}^\dag S(\tilde{\omega})^\dag \tilde{\phi}}.  \]
\end{proposition}
\begin{proof}
We begin by writing
    \begin{align}
        \|[S(\tilde{\omega}), f]\|_\rho^2&= \trace{ \rho^{1/4} (S(\tilde{\omega}) \phi - \phi S(\tilde{\omega})) \rho^{1/2} (\phi^\dag S(\tilde{\omega})^\dag -  S(\tilde{\omega})^\dag \phi^\dag) \rho^{1/4}}\nonumber\\
        &= \trace{\rho^{1/4} S(\tilde{\omega}) \phi \rho^{1/2} \phi^\dag S(\tilde{\omega})^\dag \rho^{1/4} } + \trace{\rho^{1/4} \phi S(\tilde{\omega})  \rho^{1/2} S(\tilde{\omega})^\dag  \phi^\dag\rho^{1/4} } \nonumber\\
        &\quad- \trace{\rho^{1/4} S(\tilde{\omega}) \phi \rho^{1/2} S(\tilde{\omega})^\dag \phi^\dag \rho^{1/4}} -\trace{\rho^{1/4} \phi S(\tilde{\omega}) \rho^{1/2} \phi^\dag   S(\tilde{\omega})^\dag \rho^{1/4}}.\label{eq:rho norm commutator}
    \end{align}

Using \cref{prop:operator in V omega commutation with rho} for $S(\tilde{\omega}) \in V_{\tilde{\omega}}$ and $S(\tilde{\omega})^\dag \in V_{-\tilde{\omega}}$ and the fact that $\trace{AB}= \trace{BA}$ from \cref{prop:eigenvalue multiplicities of AB and BA}, we have
\begin{align*}
    \trace{\rho^{1/4} S(\tilde{\omega}) \phi \rho^{1/2} \phi^\dag S(\tilde{\omega})^\dag \rho^{1/4} }  &= e^{-\frac{\beta\tilde{\omega}}{2}} \trace{ S(\tilde{\omega}) \rho^{1/4}\phi \rho^{1/2} \phi^\dag \rho^{1/4}S(\tilde{\omega})^\dag}  = e^{-\frac{\beta\tilde{\omega}}{2}} \trace{S(\tilde{\omega}) \tilde{\phi} \tilde{\phi}^\dag S(\tilde{\omega}) ^\dag } , \\
    \trace{\rho^{1/4} \phi S(\tilde{\omega})  \rho^{1/2} S(\tilde{\omega})^\dag  \phi^\dag\rho^{1/4} }  &=  e^{\frac{\beta\tilde{\omega}}{2}}\trace{\rho^{1/4} \phi \rho^{1/4} S(\tilde{\omega})   S(\tilde{\omega})^\dag  \rho^{1/4} \phi^\dag\rho^{1/4} } =  e^{\frac{\beta\tilde{\omega}}{2}}\trace{S(\tilde{\omega})^\dag \tilde{\phi}^\dag \tilde{\phi} S(\tilde{\omega})},\\
    \trace{\rho^{1/4} S(\tilde{\omega}) \phi \rho^{1/2} S(\tilde{\omega})^\dag \phi^\dag \rho^{1/4}} &= \trace{ S(\tilde{\omega}) \rho^{1/4} \phi \rho^{1/4} S(\tilde{\omega})^\dag \rho^{1/4} \phi^\dag \rho^{1/4}} = \trace{S(\tilde{\omega}) \tilde{\phi} S(\tilde{\omega})^\dag \tilde{\phi}^\dag},\\
    \trace{\rho^{1/4} \phi S(\tilde{\omega}) \rho^{1/2}  \phi^\dag   S(\tilde{\omega})^\dag\rho^{1/4}} &=\trace{\rho^{1/4} \phi \rho^{1/4} S(\tilde{\omega}) \rho^{1/4} \phi^\dag \rho^{1/4} S(\tilde{\omega})^\dag  } =  \trace{ \tilde{\phi} S(\tilde{\omega}) \tilde{\phi}^\dag S(\tilde{\omega})^\dag } = \trace{ S(\tilde{\omega}) \tilde{\phi}^\dag S(\tilde{\omega})^\dag \tilde{\phi}}. 
\end{align*}
Plugging each of the four equalities above into \cref{eq:rho norm commutator} gives the result. 
\end{proof}

\begin{proposition} \label{prop:crucial trace inequality}
   Let $ A, B$ be arbitrary matrices.
   Let $f$ be an arbitrary matrix, and let $g := (f f ^\dag)^{1/2}$ and $h := (f^\dag f)^{1/2}$. Then, $ \tr (A g A^\dag g) + \tr (Bh B^\dag h ) \geq \tr (A f B^\dag f^\dag) + \tr (B f^\dag A^\dag f).$

\end{proposition}
\begin{proof}
Pick an orthonormal eigenbasis $u_1, \dots, u_m$ of $f f^\dag = g^2$ with corresponding eigenvalues $\lambda_1 \geq \lambda_2 \geq \cdots\geq\lambda_k > 0 =\lambda_{k+1} =\cdots =\lambda_m = 0.$
Then, 
\[ g = \sum_{i=1}^k \sqrt{\lambda_i} u_i u_i^\dag. \]
Let $ v_i = f^\dag u_i$. If $ i\geq k+1,$ then $f^\dag u_i =0 $ since $ \|f^\dag u_i\|_2^2 =  u_i^\dag f f^\dag u_i = 0 .$ 
Thus
\[ f = \sum_{i= 1}^m u_i u_i^\dag f  = \sum_{i=1}^k u_i v_i^\dag \]
and since $\{u_i\}_{i=1}^m$ is an orthonormal eigenbasis, $ \sum_{i=1}^m u_i u_i^\dag = I$ and
\[ f^\dag f = f^\dag (\sum_{i=1}^m u_i u_i^\dag )f  = \sum_{i=1}^k v_i v_i^\dag. \]
For $i\in [k], \|v_i\|_2^2 = u_i^\dag f f^\dag u_i = \lambda_i > 0$, so we can rewrite $ f^\dag f = \sum_{i=1}^k \lambda_i \cdot \frac{v_i}{\sqrt{\lambda_i}} \cdot \frac{v_i^\dag}{\sqrt{\lambda_i}}$ and thus 
\[h =(f^\dag f)^{1/2} = \sum_{i=1}^k \sqrt{\lambda_i}  \cdot\frac{v_i}{\sqrt{\lambda_i}} \cdot \frac{v_i^\dag}{\sqrt{\lambda_i}}.\]
Hence, for $i\in[k]$, letting $\tilde{v}_i = \frac{v_i}{\sqrt{\lambda_i}}$:
\begin{equation}\label{eq:svd decomposition of f g h}
 g = \sum_{i=1}^k \sqrt{\lambda_i} u_i u_i^\dag, \quad f = \sum_{i=1}^k \sqrt{\lambda_i} u_i \tilde{v}_i^\dag, \quad h = \sum_{i=1}^k \sqrt{\lambda_i} \tilde{v}_i \tilde{v}_i^\dag.
\end{equation}
Note also that the $\tilde{v}_i $ have unit norm and are pairwise orthogonal, since $v_i^\dag v_j = u_i^\dag f f ^\dag u_j =0$ if $i\neq j$.
We can rewrite
\begin{align*}
         &\tr (A g A^\dag g) + \tr (Bh B^\dag h ) - \tr (A f B^\dag f^\dag) - \tr (B f^\dag A^\dag f)\\
         &\qquad= \sum_{i, j} \sqrt{\lambda_i\lambda_j} u_i^\dag A u_j u_j^\dag A^\dag u_i +   \sum_{i, j} \sqrt{\lambda_i\lambda_j} \tilde{v}_i^\dag B \tilde{v}_j \tilde{v}_j^\dag B^\dag \tilde{v}_i  - \sum_{i,j}  \sqrt{\lambda_i\lambda_j}  u_i^\dag A u_j \tilde{v}_j^\dag B^\dag \tilde{v}_i - \sum_{i,j} \sqrt{\lambda_i\lambda_j}   \tilde{v}_i^\dag B \tilde{v}_j u_j^\dag A^\dag u_i \\
         &\qquad= \sum_{i,j} \sqrt{\lambda_i\lambda_j}  | u_i^\dag A u_j- \tilde{v}_i^\dag B \tilde{v}_j |^2  \geq 0. 
         \end{align*}
\end{proof}

\begin{proof}[Proof of \cref{thm:comparison of dirichlet forms}]

We write the Dirichlet forms of $f,g,h$ as sum of commutators as in \cref{prop:dirichlet form as sum of commutators}. For each Bohr frequency $\tilde{\omega}$, introduce 
\begin{align*}
T(\tilde{\omega}) :=\| [S(\tilde{\omega}), f ]\|_\rho^2 +  \| [S^\dag(-\tilde{\omega}), f ]\|_\rho^2 - \frac{1}{2} (  \| [S(\tilde{\omega}), g ]\|_\rho^2 +  \| [S^\dag(-\tilde{\omega}), g ]\|_\rho^2  + \| [S(\tilde{\omega}), h ]\|_\rho^2 +  \| [S^\dag(-\tilde{\omega}), h ]\|_\rho^2 ).
\end{align*}
Using \cref{prop:dirichlet form as sum of commutators} and the fact that $ \tilde{G}(\tilde{\omega}) = G(\tilde{\omega}) e^{\frac{\beta \tilde{\omega}}{2}} =\tilde{G}(-\tilde{\omega})$, we can rewrite
\[2 \mathcal{E}_\mathcal{L}(f) =\sum_{S\in \mathcal{S},\tilde{\omega}} \tilde{G}(\tilde{\omega}) \| [S(\tilde{\omega}), f ]\|_\rho^2   = \frac{1}{2} ( \sum_{S\in \mathcal{S},\tilde{\omega}} \tilde{G}(\tilde{\omega}) \| [S(\tilde{\omega}), f ]\|_\rho^2 + \sum_{S^\dag\in \mathcal{S^\dag},-\tilde{\omega}} \tilde{G}(\tilde{\omega}) \| [S^\dag(-\tilde{\omega}), f ]\|_\rho^2) \]
hence
\begin{align*}
     2 \mathcal{E}_\mathcal{L}(f) - (\mathcal{E}_\mathcal{L}(g) + \mathcal{E}_\mathcal{L}(h)) =  \frac{1}{2}\sum_{\tilde{\omega}} \tilde{G}(\tilde{\omega})T(\tilde{\omega})
\end{align*}

We will show that $T(\tilde{\omega})\ge0$ and then conclude that
\[2 \mathcal{E}_\mathcal{L}(f) - (\mathcal{E}_\mathcal{L}(g) + \mathcal{E}_\mathcal{L}(h)) \geq 0\]
Let $\tilde{f} = \rho^{1/4} f\rho^{1/4}, \tilde{g} = \rho^{1/4} g\rho^{1/4}, \tilde{h} = \rho^{1/4} h\rho^{1/4}$ 
as in \cref{prop:commutator decomposition}. By \cref{prop:operator in V omega commutation with rho}, since $f\in V_\omega$ and $g,h\in V_0$ and are Hermitian, 
\begin{align*}
    \tilde{g}^2 &=\tilde{g} \tilde{g}^\dag = \rho^{1/4} g \rho^{1/2} g^\dag \rho^{1/4} = \rho^{1/2} g^2 \rho^{1/2} = e^{-\frac{\beta\omega}{2}} \rho^{1/2} f f^\dag  \rho^{1/2} = \rho^{1/4} f \rho^{1/2} f^\dag \rho^{1/2} = \tilde{f}\tilde{f}^\dag,\\
    \tilde{h}^2 &= \tilde{h} \tilde{h}^\dag = \rho^{1/4} h \rho^{1/2} h^\dag \rho^{1/4} = \rho^{1/2} h^2 \rho^{1/2} = e^{\frac{\beta\omega}{2}} \rho^{1/2} f^\dag f  \rho^{1/2} = \rho^{1/4} f^\dag \rho^{1/2} f \rho^{1/2} = \tilde{f}^\dag\tilde{f}. 
\end{align*}
Using \cref{prop:commutator decomposition}, and note that $S^\dag(-\tilde{\omega}) =S(\tilde{\omega})^\dag,$ we can rewrite
\begin{align*}
     \|[S(\tilde{\omega}), f]\|_\rho^2 &= e^{-\frac{\beta \tilde{\omega}}{2}} 
 \trace{S(\tilde{\omega}) \tilde{f} \tilde{f}^\dag S(\tilde{\omega})^\dag } + e^{\frac{\beta \tilde{\omega}}{2}} 
\trace{S(\tilde{\omega})^\dag \tilde{f}^\dag \tilde{f} S(\tilde{\omega})} -  \trace{S(\tilde{\omega}) \tilde{f} S(\tilde{\omega})^\dag \tilde{f}^\dag} -  \trace{S(\tilde{\omega}) \tilde{f}^\dag S(\tilde{\omega})^\dag \tilde{f}}.\\
  \|[S^\dag(-\tilde{\omega}), f]\|_\rho^2
  &= e^{\frac{\beta \tilde{\omega}}{2}} 
 \trace{S(\tilde{\omega})^\dag \tilde{f} \tilde{f}^\dag S(\tilde{\omega}) } + e^{-\frac{\beta \tilde{\omega}}{2}} 
\trace{S(\tilde{\omega}) \tilde{f}^\dag \tilde{f} S(\tilde{\omega})^\dag} -  \trace{S(\tilde{\omega})^\dag \tilde{f} S(\tilde{\omega}) \tilde{f}^\dag} -  \trace{S(\tilde{\omega})^\dag \tilde{f}^\dag S(\tilde{\omega}) \tilde{f}}.\\
\|[S(\tilde{\omega}), g]\|_\rho^2 &= e^{-\frac{\beta \tilde{\omega}}{2}} 
 \trace{S(\tilde{\omega}) \tilde{g} \tilde{g}^\dag S(\tilde{\omega})^\dag } + e^{\frac{\beta \tilde{\omega}}{2}} 
\trace{S(\tilde{\omega})^\dag \tilde{g}^\dag \tilde{g} S(\tilde{\omega})} -  \trace{S(\tilde{\omega}) \tilde{g} S(\tilde{\omega})^\dag \tilde{g}^\dag} -  \trace{S(\tilde{\omega}) \tilde{g}^\dag S(\tilde{\omega})^\dag \tilde{g}}\\
&= e^{-\frac{\beta \tilde{\omega}}{2}} 
 \trace{S(\tilde{\omega}) \tilde{f} \tilde{f}^\dag S(\tilde{\omega})^\dag } + e^{\frac{\beta \tilde{\omega}}{2}} 
\trace{S(\tilde{\omega})^\dag \tilde{f} \tilde{f}^\dag S(\tilde{\omega})} -  \trace{S(\tilde{\omega}) \tilde{g} S(\tilde{\omega})^\dag \tilde{g}} -  \trace{S(\tilde{\omega}) \tilde{g} S(\tilde{\omega})^\dag \tilde{g}}\\
&= e^{-\frac{\beta \tilde{\omega}}{2}} 
 \trace{S(\tilde{\omega}) \tilde{f} \tilde{f}^\dag S(\tilde{\omega})^\dag } + e^{\frac{\beta \tilde{\omega}}{2}} 
\trace{S(\tilde{\omega})^\dag \tilde{f} \tilde{f}^\dag S(\tilde{\omega})} -  2 \trace{S(\tilde{\omega}) \tilde{g} S(\tilde{\omega})^\dag \tilde{g}} 
\intertext{where we use $ \tilde{g} =\tilde{g}^\dag$ and $ \tilde{g}^2 = \tilde{f} \tilde{f}^\dag.$  Similarly,}
\|[S^\dag(-\tilde{\omega}), g]\|_\rho^2 &= e^{\frac{\beta \tilde{\omega}}{2}} 
 \trace{S(\tilde{\omega})^\dag \tilde{g} \tilde{g}^\dag S(\tilde{\omega}) } + e^{-\frac{\beta \tilde{\omega}}{2}} 
\trace{S(\tilde{\omega}) \tilde{g}^\dag \tilde{g} S(\tilde{\omega})^\dag} -  \trace{S(\tilde{\omega})^\dag \tilde{g} S(\tilde{\omega}) \tilde{g}^\dag} -  \trace{S(\tilde{\omega})^\dag \tilde{g}^\dag S(\tilde{\omega}) \tilde{g}}\\
&=e^{\frac{\beta \tilde{\omega}}{2}} 
 \trace{S(\tilde{\omega})^\dag \tilde{f} \tilde{f}^\dag S(\tilde{\omega}) } + e^{-\frac{\beta \tilde{\omega}}{2}} 
\trace{S(\tilde{\omega}) \tilde{f} \tilde{f}^\dag S(\tilde{\omega})^\dag} -  2 \trace{ S(\tilde{\omega}) \tilde{g} S(\tilde{\omega})^\dag \tilde{g}} \\
\intertext{since by \cref{prop:eigenvalue multiplicities of AB and BA},  $ \trace{ (S(\tilde{\omega}) \tilde{g} ) (S(\tilde{\omega})^\dag \tilde{g} ) } =  \trace{( S(\tilde{\omega})^\dag \tilde{g} ) (S(\tilde{\omega})\tilde{g}) }.$ Analogously, using $ \tilde{h} =\tilde{h}^\dag$ and $ \tilde{h}^2 = \tilde{f}^\dag \tilde{f},$ we obtain:}
\|[S(\tilde{\omega}), h]\|_\rho^2 &= e^{-\frac{\beta \tilde{\omega}}{2}} 
 \trace{S(\tilde{\omega}) \tilde{f}^\dag \tilde{f} S(\tilde{\omega})^\dag } + e^{\frac{\beta \tilde{\omega}}{2}} 
\trace{S(\tilde{\omega})^\dag  \tilde{f}^\dag \tilde{f} S(\tilde{\omega})} -  2\trace{S(\tilde{\omega}) \tilde{h} S(\tilde{\omega})^\dag \tilde{h}}\\ 
\|[S^\dag(-\tilde{\omega}), h]\|_\rho^2 &= e^{\frac{\beta \tilde{\omega}}{2}} 
 \trace{S(\tilde{\omega})^\dag \tilde{f}^\dag \tilde{f} S(\tilde{\omega}) } + e^{-\frac{\beta \tilde{\omega}}{2}} 
\trace{S(\tilde{\omega}) \tilde{f}^\dag \tilde{f} S(\tilde{\omega})^\dag} -  2\trace{S(\tilde{\omega}) \tilde{h} S(\tilde{\omega})^\dag \tilde{h}}
\end{align*}

Let
\begin{align*}
 M_1 :&=  e^{-\frac{\beta \tilde{\omega}}{2}} 
 \trace{S(\tilde{\omega}) \tilde{f} \tilde{f}^\dag S(\tilde{\omega})^\dag } + e^{\frac{\beta \tilde{\omega}}{2}} 
 \trace{S(\tilde{\omega})^\dag \tilde{f} \tilde{f}^\dag S(\tilde{\omega}) } \\
 M_2 :&=  e^{\frac{\beta \tilde{\omega}}{2}} 
\trace{S(\tilde{\omega})^\dag \tilde{f}^\dag \tilde{f} S(\tilde{\omega})}  + e^{-\frac{\beta \tilde{\omega}}{2}} 
\trace{S(\tilde{\omega}) \tilde{f}^\dag \tilde{f} S(\tilde{\omega})^\dag}
\end{align*}
Observe that \cref{prop:eigenvalue multiplicities of AB and BA},  $\trace{S(\tilde{\omega}) \tilde{f} S(\tilde{\omega})^\dag \tilde{f}^\dag} = \trace{(S(\tilde{\omega}) \tilde{f}) (S(\tilde{\omega})^\dag \tilde{f}^\dag)} =  \trace{(S(\tilde{\omega})^\dag \tilde{f}^\dag)(S(\tilde{\omega}) \tilde{f}) } = \trace{S(\tilde{\omega})^\dag \tilde{f}^\dag S(\tilde{\omega}) \tilde{f}} ,$ and similarly,  $\trace{S(\tilde{\omega}) \tilde{f}^\dag S(\tilde{\omega})^\dag \tilde{f}} = \trace{S(\tilde{\omega})^\dag \tilde{f} S(\tilde{\omega}) \tilde{f}^\dag},$ so we can rewrite
\begin{align*}
     \|[S(\tilde{\omega}), f]\|_\rho^2 + \|[S^\dag(-\tilde{\omega}), f]\|_\rho^2 &= M_1 + M_2 - 2 \trace{S(\tilde{\omega}) \tilde{f} S(\tilde{\omega})^\dag \tilde{f}^\dag} - 2 \trace{S(\tilde{\omega}) \tilde{f}^\dag S(\tilde{\omega})^\dag \tilde{f}} \\
      \|[S(\tilde{\omega}), g]\|_\rho^2 + \|[S^\dag(-\tilde{\omega}), g]\|_\rho^2 &= 2 M_1- 4 \trace{S(\tilde{\omega}) \tilde{g} S(\tilde{\omega})^\dag \tilde{g}}\\
       \|[S(\tilde{\omega}), h]\|_\rho^2 + \|[S^\dag(-\tilde{\omega}), h]\|_\rho^2 &= 2 M_2 - 4 \trace{S(\tilde{\omega}) \tilde{h} S(\tilde{\omega})^\dag \tilde{h}}
\end{align*}
thus
\begin{align*}
    T(\tilde{\omega}) =     
    2 (\trace{S(\tilde{\omega}) \tilde{g} S(\tilde{\omega})^\dag \tilde{g}} + \trace{S(\tilde{\omega}) \tilde{h} S(\tilde{\omega})^\dag \tilde{h}} -\trace{S(\tilde{\omega}) \tilde{f} S(\tilde{\omega})^\dag \tilde{f}^\dag} -\trace{S(\tilde{\omega}) \tilde{f}^\dag S(\tilde{\omega})^\dag \tilde{f}}) \geq 0
\end{align*}
where the inequality is due to \cref{prop:crucial trace inequality}, and the fact that $ \tilde{g} =c$ and $\tilde{h} =  (\tilde{f}^\dag \tilde{f})^{1/2} $.
\end{proof}
\begin{theorem}[Comparing variances]\label{thm:comparison of variances}
    Take $\omega\neq0$ and let $D\equiv D_{H,\omega}$ be the length of the longest difference-$\omega$ arithmetic sequence in the spectrum of the Hamiltonian $H$. For any $f \in V_\omega$, and $g,h$ as defined in \cref{eq:hermitianization of f}. 
    Then,
\[\Var(g) + \Var(h) \geq \frac{1}{D}\Var(f).\]
\end{theorem}
\begin{proof}

    Note that $g\in V_0$ and thus $g$ commutes with $H$ (see \cref{prop:observable in V0 commute with $H$}) and is Hermitian, so using the same argument as in \cref{eq:svd decomposition of f g h} in the proof of \cref{prop:crucial trace inequality},
    there exists an orthonormal eigenbasis $\set{u_i}$ of $H$ with corresponding eigenvalues $H(u_i) = \bra{u_i} H\ket{u_i}$ such that 
    \begin{align} \label{eq:expansion of f g h}
      g = e^{\frac{\beta\omega }{4}} \sum_{i=1}^k \lambda_i \ket{u_i} \bra{u_i},  && f = \sum_{i=1}^k \lambda_i \ket{u_i} \bra{v_i}, && h = e^{-\frac{\beta\omega }{4}} \sum_{i=1}^k \lambda_i \ket{v_i}\bra{v_i}  
    \end{align}
where $\lambda_i \in \R_{>0}$ for each $ i\in [k]$ and $\ket{v_i} = \frac{f^\dag \ket{u_i}}{\lambda_i}.$ Observe that the vectors $v_i$ are pairwise orthogonal, since $v_i^\dag v_j = (\lambda_i \lambda_j)^{-1}\cdot  u_i^\dag f f ^\dag u_j =  (\lambda_i \lambda_j)^{-1} u_i^\dag g^2 u_j^\dag = 0$ if $i\neq j.$ In addition, for each $i\in [k]$, the vector $v_i$ has unit norm. Note also that $v_i$ is an eigenvector of $H$
 with eigenvalues $H(v_i) =\bra{v_i}H \ket{v_i} =H(u_i) - \omega,$ since $f = f(\omega) = \sum_{E_1,E_2 :E_1 -E_2 =\omega} \Pi_{E_1} f \Pi_{E_2} $ and $\Pi_{E_1} \ket{u_i}=0$ unless $E_1= H(u_i),$ we have
 \[\ket{v_i} =  \lambda_i^{-1} f^\dag \ket{u_i}  =  \lambda_i^{-1} \Pi_{H(u_i)-\omega} f^\dag \Pi_{H(u_i)} \ket{u_i} =  \Pi_{H(u_i)-\omega} \ket{v_i}  \]
 so
 \[ H \ket{v_i} =  (H(u_i)-\omega) \ket{v_i}.\]

Let $\rho(u_i) = \bra{u_i}\rho\ket{u_i}$ and $\rho(v_i) = \bra{v_i}\rho \ket{v_i}$ be the eigenvalue of $\rho$ corresponding to eigenvectors $u_i$ and $v_i$. Then, $ \rho(v_i) = \rho(u_i) e^{\beta \omega}.$
     We write
\begin{equation}\label{eq:variance of f}
    \Var(f) = \sum_{i=1}^k \sqrt{\rho(u_i)\rho(v_i)} \lambda_i^2 =  e^{\frac{\beta\omega }{2}}  \sum_{i=1}^k \rho(u_i) \lambda_i^2  = e^{-\frac{\beta\omega }{2}}  \sum_{i=1}^k \rho(v_i) \lambda_i^2 
\end{equation}
    \begin{equation} \label{eq:holder ineq to simplify for g}
        \Var(g) =e^{\frac{\beta\omega }{2}} 
\left( \sum_{i=1}^k \rho(u_i) \lambda_i^2 -\left(\sum_{i=1}^k \rho(u_i) \lambda_i \right)^2 \right) \geq \left(1- \sum_{i=1}^k \rho(u_i)\right) \left( e^{\frac{\beta\omega }{2}}  \sum_{i=1}^k \rho(u_i) \lambda_i^2\right) =\left(1- \sum_{i=1}^k \rho(u_i)\right) \Var(f)
    \end{equation}
    where the inequality is due to H\"older's inequality.
    Similarly, we have
    \begin{equation} \label{eq:holder ineq to simplify for h}
        \Var(h) =e^{-\frac{\beta\omega }{2}} 
\left( \sum_{i=1}^k \rho(v_i) \lambda_i^2 -\left(\sum_{i=1}^k \rho(v_i) \lambda_i \right)^2 \right) \geq \left(1- \sum_{i=1}^k \rho(v_i)\right) \Var(f). 
\end{equation}
    
    We first consider the case $\beta \omega \geq 0.$ 
    We show that 
    \begin{equation}\label{ineq:eigenvalue lower bound for variance comparison}
          1 - \sum_{i=1}^k \rho(u_i) \geq \frac{1}{D}
    \end{equation}
    which together with \cref{eq:holder ineq to simplify for g} gives the desired conclusion that $\Var(g) \geq \frac{1}{D}\Var(f)$.


   Let the distinct values appearing in the sequence $ (H(u_i))_{i=1}^k$ be $\alpha_1, \dots, \alpha_q.$ 
    Partition the set $Q := \set{\alpha_1, \dots, \alpha_q}$ into a family $\mathcal{I}$ of (disjoint) maximally long arithmetic sequences where in each sequence, the difference between two consecutive terms is $\omega.$ Such a decomposition can be obtained by iteratively removing the maximally long arithmetic progression from the set $Q.$
    
    We rewrite
    \[Q = \bigcup_{i\in \mathcal{I}} \set{\alpha_{i,1}, \dots, \alpha_{i,d_i}}   \]
    where for each $i\in \mathcal{I},$ $ \alpha_{i,1}, \dots, \alpha_{i,d_i}$ satisfies $ \alpha_{i,j} = \alpha_{i,1}-\omega(j-1)$ for all $j\in [d_i].$
Note that $ \alpha_{i,d_i+1}:= \alpha_{i,d_i}-\omega$ is also an eigenvalue of $\rho,$ since if $ u_{\hat{j}}$ is an eigenvector of $H$ with eigenvalue $\alpha_{i,d_i}$ that appears in \cref{eq:expansion of f g h}, then $f^\dag u_{\hat{j}}$ is an eigenvector of $H$ with eigenvalue $\alpha_{i,d_i}-\omega. $ Thus the spectrum of $H$ contains the length-$(d_i+1)$ arithmetic progression $ \alpha_{i,1}, \dots, \alpha_{i, {d_i+1}}$ with difference between consecutive terms being $\omega.$ Note also that $ \alpha_{i,d_i+1}\not\in Q,$ otherwise this contradicts the maximality of $ \alpha_{i,1}, \dots, \alpha_{i,d_i}.$

For an eigenvalue $\alpha$ of $H, $ let $\hat{\alpha} := \frac{\exp(-\beta \alpha)}{\trace{\exp(-\beta H)}}\geq 0$ be the corresponding eigenvalue of $ \rho,$ $m(\alpha)$ be its multiplicity in the spectrum of $H$, and $r(\alpha)$ be its multiplicity in the sequence $ (H(u_i))_{i=1}^k.$ Obviously, $r(\alpha)\leq m(\alpha).$

We show, for each $i\in \mathcal{I},$
\begin{equation} \label{ineq:eigenvalue sum bound by arithmetic sequence}
     m( \alpha_{i, {d_i+1}})\hat{\alpha}_{i, {d_i+1}}  + \sum_{j=1}^{d_i} (m(\alpha_{i,j}) -r(\alpha_{i,j}))\hat{\alpha}_{i,j} \geq \frac{1}{d_i+1}\sum_{j=1}^{d_i+1} m(\alpha_{i,j}) \hat{\alpha}_{i,j} \geq \frac{1}{D}\sum_{j=1}^{d_i+1} m(\alpha_{i,j}) \hat{\alpha}_{i,j} 
\end{equation}
where the last inequality is because $ d_i+1\leq D,$ due to the assumption on the spectrum of $H$ and the fact that $H$ contains the length-$(d_i+1)$ arithmetic progression $ \alpha_{i,1}, \dots, \alpha_{i, {d_i+1}}$ with difference between consecutive terms being $\omega.$ 

Before showing \cref{ineq:eigenvalue sum bound by arithmetic sequence}, we check that it implies the desired inequality \cref{ineq:eigenvalue lower bound for variance comparison}.
       \begin{align*} 
     1- \sum_{j=1}^k \rho(u_j) &=_{(1)} \sum_\alpha m(\alpha) \hat{\alpha} - \sum_{i\in \mathcal{I}, j \in [d_i]} r(\alpha_{i,j})\hat{\alpha}_{i,j}  \\
     &=\sum_{i\in \mathcal{I} } \sum_{j=1}^{d_i+1} m(\alpha_{i,j} )\hat{\alpha}_{i,j} + \sum_{\alpha \not\in \set{\alpha_{i,j}:i\in \mathcal{I}, j \in [d_{i}+1] }} m(\alpha) \hat{\alpha}  -\sum_{i \in \mathcal{I} } \sum_{j=1}^{d_i} r(\alpha_{i,j}) \hat{\alpha}_{i,j} \\
     &= \sum_{i \in \mathcal{I} } \left( m( \alpha_{i, {d_i+1}})\hat{\alpha}_{i, {d_i+1}}  + \sum_{j=1}^{d_i} (m(\alpha_{i,j}) -r(\alpha_{i,j}))\hat{\alpha}_{i,j} \right)  + \sum_{\alpha \not\in \set{\alpha_{i,j}:i\in \mathcal{I}, j \in [d_{i}+1] }} m(\alpha) \hat{\alpha} \\
     &\geq_{(2)} \frac{1}{D} \sum_{\alpha}m(\alpha) \hat{\alpha} \\
     &= \frac{1}{D}
    \end{align*} 
    where (1) and the final equality is by $1=\trace{\rho} = \sum_{\alpha}m(\alpha) \hat{\alpha} ,$ and (2) is by \cref{ineq:eigenvalue sum bound by arithmetic sequence}.

Now, we show \cref{ineq:eigenvalue sum bound by arithmetic sequence}.    
    Recall that for all $j\in [d_i]$, $r(\alpha_{i,j}) \leq m(\alpha_{i,j}). $  Furthermore, for all $j\in [d_i]$, $r(\alpha_{i,j}) \leq m(\alpha_{i,j+1}).$ Indeed, fix one such $j.$ Let $u_{h_1}, \dots, u_{h_r}$ be the eigenvectors of $H$ from \cref{eq:expansion of f g h} with eigenvalue $ \alpha_{i,j}$. Then $ f^\dag u_{h_1}, \dots, f^\dag u_{h_r} $  are pairwise orthogonal (and thus linearly independent) eigenvectors of $H$ with eigenvalue $ \alpha_{i,j} - \omega = \alpha_{i,{j+1}},$ hence $ r(\alpha_{i,j}) \leq m(\alpha_{i,j+1}).$

    For $j^*= \argmax_{j\in [d_i+1]} m(\alpha_{i,j})  \hat{\alpha}_{i,j},$ we have
    \begin{align*}
    &m(\alpha_{i,d_i+1})\hat{\alpha}_{i, d_i+1}+ \sum_{j=1}^{d_i} (m(\alpha_{i,j}) - r(\alpha_{i,j})) \hat{\alpha}_{i,j} \\
        &\qquad = \sum_{j=1}^{j^*-1} (m(\alpha_{i,j}) - r(\alpha_{i,j})) \hat{\alpha}_{i,j}+ m(\alpha_{i,j^*}) \hat{\alpha}_{i,j^*} + \sum_{j=j^*}^{d_i} (m(\alpha_{i,j+1}) \hat{\alpha}_{i,j+1} -r(\alpha_{i,j}) \hat{\alpha}_{i,j}) \\
        &\qquad \geq m(\alpha_{i,j^*}) \hat{\alpha}_{i,j^*}
    \end{align*}
    where the inequality is due to $ (m(\alpha_{i,j}) -r(\alpha_{i,j}))\hat{\alpha}_{i,j}\geq 0$, and 
    \[m(\alpha_{i,j+1}) \hat{\alpha}_{i,j+1} -r(\alpha_{i,j}) \hat{\alpha}_{i,j} = \hat{\alpha}_{i,j}(e^{\beta \omega} m(\alpha_{i,j+1}) - r(\alpha_{i,j}) ) \geq \hat{\alpha}_{i,j}(m(\alpha_{i,j+1}) - r(\alpha_{i,j}) )\geq 0, \]
    where we use the assumption $\beta \omega \geq 0$ and the fact that $ m(\alpha_{i,j+1}) \geq r(\alpha_{i,j}).$ Thus we have obtained
    \[ m(\alpha_{i,d_i+1}) \hat{\alpha}_{i,d_i+1}+ \sum_{j=1}^{d_i} (m(\alpha_{i,j}) - r(\alpha_{i,j})) \hat{\alpha}_{i,j} \geq m(\alpha_{i,j^*}) \hat{\alpha}_{i,j^*} \geq \frac{1}{d_i+1} \sum_{j=1}^{d_i+1} m(\alpha_{i,j})\hat{\alpha}_{i,j}. \]
and conclude that $ \Var(g) \geq\frac{1}{D}\Var(f),$ which implies  $\Var(g) + \Var(h) \geq \frac{1}{D} \Var(f).$

    The case $ \beta \omega \leq 0 $ follows similarly, where we prove $\Var(h) \geq \frac{1}{D} \Var(f).$ Thus, in either case, we have 
    \[\Var(g) + \Var(h) \geq \frac{1}{D} \Var(f).\qedhere\]
 \end{proof}
We have the following alternative comparison between the variances.
\begin{lemma}\label{lem:alternative comparison of variances}
    For any $f \in V_\omega$, and $g,h$ as defined in \cref{eq:hermitianization of f}. Then,
\[\Var(g) + \Var(h) \geq (1-e^{-|\beta \omega|})\Var(f).\]
\end{lemma}
\begin{proof}
   From the proof of \cref{thm:comparison of variances}, we have
    \[ \Var(g) \geq \left(1- \sum_{i=1}^k \rho(u_i)\right) \Var(f)\quad \text{and} \quad \Var(h) \geq \left(1- \sum_{i=1}^k \rho(v_i)\right) \Var(f)\]
    where for $i\in [k],$ $u_i, v_i$ are eigenvectors of the stationary state $\rho$ with corresponding eigenvalues $ \rho(u_i), \rho(v_i)$ respectively, and $ \rho(v_i) = \rho(u_i)e^{\beta\omega}. $ Moreover, $ \set{u_i}_{i=1}^k$ are linearly independent and $\set{v_i}_{i=1}^k$ are linearly independent, so $ \max \set{\sum_{i=1}^k\rho(u_i),\sum_{i=1}^k\rho(v_i)} \leq \trace{\rho} = 1$.
    If $\beta \omega \geq 0$ then 
    \[e^{\beta\omega } \sum_{i=1}^k\rho(u_i) 
 = \sum_{i=1}^k\rho(v_i) \leq 1 \]
 so \[ \Var(g) = \left(1-\sum_{i=1}^k\rho(u_i) \right) \Var(f) \geq \left(1- e^{-|\beta \omega|}\right) \Var(f).\]
 Similarly, if $\beta \omega \leq 0$ then
 \[e^{-\beta\omega } \sum_{i=1}^k\rho(v_i)=\sum_{i=1}^k\rho(u_i)\leq 1\Rightarrow \Var(h) = \left(1-\sum_{i=1}^k\rho(v_i) \right) \Var(f) \geq \left(1- e^{-|\beta \omega|}\right) \Var(f). \qedhere\]
\end{proof}
\begin{proof}[Proof of \cref{thm:compare spectral gap V omega and V 0}]
    Fix an arbitrary $\omega\neq 0$ and an arbitrary $f \in V_\omega.$ Let $g,h $ be as defined in \cref{thm:comparison of dirichlet forms}; then, $g, h \in V_0.$ 
     We have
    \[ 2 \mathcal{E}(f) \geq_{(1)} \mathcal{E}(g)  + \mathcal{E}(h) \geq_{(2)} \lambda_{\mathcal{L},0} (\Var(g) + \Var(h)) \geq_{(3)} \lambda_{\mathcal{L},0}\max(\frac{1}{D}, 1 -e^{-|\beta \omega|}) \Var(f) \]
    where (1) is by \cref{thm:comparison of dirichlet forms}, (2) by definition of $\lambda_{\mathcal{L},0},$ and (3) by \cref{thm:comparison of variances,lem:alternative comparison of variances}. Thus
\[ \lambda_{\mathcal{L},\omega} =\min_{f \in V_\omega} \frac{\mathcal{E}(f)}{\Var(f)} \geq \frac{1}{2}\max(\frac{1}{D}, 1 -e^{-|\beta \omega|}) \lambda_{\mathcal{L},0}. \qedhere\]
\end{proof}
\begin{proof}[Proof of \cref{thm:general comparison theorem dependent only on length of arithmetic sequence}]
    By \cref{prop:relate spectral gap of Lindbladian to spectral gap in subspace Vomega},
    \[ \lambda_{\mathcal{L}, 0}\geq \lambda_{\mathcal{L}} = \min_\omega \lambda_{\mathcal{L}, \omega} \geq \frac{1}{2D} \lambda_{\mathcal{L},0}\]
    where the last inequality follows from: for all $\omega \neq 0$, $\lambda_{\mathcal{L},\omega} \geq \frac{1}{2D} \lambda_{\mathcal{L},0}$, by \cref{thm:compare spectral gap V omega and V 0}.
\end{proof}

%% file: spectral_gap_in_V0_vs_classical_chain.tex
\section{Relating quantum to classical spectral gap}\label{sec:quantum to classical spectral gap}
In this section, we prove \cref{prop:classical chain V0,thm:quantum vs classical spectral gap comparison}, which relate the quantum spectral gap to the spectral gap of certain classical Markov generators.

First, we prove \cref{prop:classical chain V0}. By \cref{prop:compare spectral gap with hermitian spectral gap in V0}, $\lambda_{\mathcal{L},0} =\lambda_{\mathcal{L},0}^{\mathrm{H}} ,$ so we only need to show the following.
\begin{proposition}
  For any Hamiltonian $H$ and Davies generator $\mathcal{L},$
  \begin{equation}\label{eq:V0 vs classical} \lambda_{\mathcal{L},0}^{\mathrm{H}} = \min_{U} \lambda_{\mathcal{L},U,\mathrm{cl}}\end{equation}
  where the minimum is taken over the orthonormal eigenbasis $U$ of the Hamiltonian $H,$ and $\lambda_{\mathcal{L},U,\mathrm{cl}}$ is the classical spectral gap associated with $\mathcal{L}$ and the orthonormal eigenbasis $U,$ as defined in \cref{def:classical Markov generator}.

\end{proposition}
\begin{proof}
Consider an arbitrary Hermitian $f\in V_0.$ 
Note that
$f$ commutes with $H$ and is Hermitian, so $f$ and $H$ are simultaneously diagonalizable by a certain orthonormal eigenbasis $U$ of $H.$ Let $u_i$ be the $i$-th eigenvector in $U,$ $\pi(u_i) = \bra{u_i} \rho \ket{u_i}$ and write $f = \sum f_i u_i u_i^\dag$ 
with $ f_i \in \R.$
Define $F: U\to \R$ where $F(u_i) =f_i.$ We will show that the Dirichlet form of $f$ with respect to $\mathcal{L}$ (respectively, the variance of $f$) is equal to the Dirichlet form of $F$ (respectively, the variance of $F$) with respect to the classical chain $P_U$ defined in \cref{def:classical Markov generator}. Concretely,
\[\mathcal{E}_\mathcal{L} (f)= \mathcal{E}_{\mathcal{L}, U,\mathrm{cl}} (F) \quad \text{and} \quad  \Var_\rho (f) = \Var_\pi (F).\]
Hence,
\[\lambda_{\mathcal{L},0}^{\mathrm{H}} = \min_{U} \min_{\substack{f =\sum_i  f_iu_i u_i^\dag, f_i\in \R\forall i}} \frac{\mathcal{E}_\mathcal{L} (f)}{\Var_\rho (f) } =  \min_{U} \min_{F: U\to \R} \frac{ \mathcal{E}_{\mathcal{L}, U,\mathrm{cl}} (F)}{\Var_\pi (F)} = \min_U \lambda_{\mathcal{L},U,\mathrm{cl}} \]
where the outside minimum is taken over orthonormal eigenbases $U$ of $H.$

Note that $\{\ket{u_i}\bra{u_j}\}_{i,j\in [N]}$ forms an orthonormal basis of $\C^{N\times N}$ with respect to the inner product $\langle A, B\rangle = \trace{A B^\dag}.$ Note that this is distinct from the KMS-product $\langle A, B\rangle_\rho = \trace{\rho^{1/2} A \rho^{1/2} B^\dag}.$ Since $\mathcal{L}$ is a linear map, we can write
\[\mathcal{L} (f) = \sum_{i,j,k}  f_i \langle \mathcal{L} (\ket{u_i}\bra{u_i}),\ket{u_j}\bra{u_k}\rangle \ket{u_j}\bra{u_k}. \]
We can also directly compute that for $i\neq j,$ and $E_i:=\bra{u_i} H \ket{u_i}, E_j: = \bra{u_j} H \ket{u_j}$,
\[\langle \mathcal{L} (\ket{u_j}\bra{u_j}),\ket{u_i}\bra{u_i}\rangle = G(E_j - E_i) \sum_{S\in \mathcal{S}}|\bra{u_i} S \ket{u_j} |^2 =P_{\mathcal{L},U}[u_i \to u_j] \]
and 
\begin{equation}\label{eq:return probability for MC}
    -\langle \mathcal{L} (\ket{u_i}\bra{u_i}),\ket{u_i}\bra{u_i} \rangle =\sum_{j:j\neq i}  \langle \mathcal{L} (\ket{u_j}\bra{u_j}),\ket{u_i}\bra{u_i}\rangle =\sum_{j:j\neq i} P_{\mathcal{L},U}[u_i\to u_j]. 
\end{equation}


\begin{equation} \label{eq:relate to classical spectral gap helper}
\begin{split}
    -\mathcal{E}_{\mathcal{L}} (f) =\langle \mathcal{L}(f),  f\rangle_\rho 
    &= \langle \sum_{i,j,k}  f_i \langle \mathcal{L} (\ket{u_i}\bra{u_i}),\ket{u_j}\bra{u_k}\rangle \ket{u_j}\bra{u_k}, \sum_{r} f_r \ket{u_r}\bra{u_r} \rangle_\rho\\
    &= \sum_{i,j,k,r} f_i f_r  \langle \mathcal{L} (\ket{u_i}\bra{u_i}),\ket{u_j}\bra{u_k}\rangle \langle  \ket{u_j}\bra{u_k}, \ket{u_r}\bra{u_r} \rangle_\rho
\end{split}
\end{equation}

The reversibility condition of the classical Markov generator follows from the reversibility of $\mathcal{L}$ (see \cref{eq:kms reversible condition}). Indeed,
\begin{equation}\label{eq:reversibility of classical MC}
   \pi(u_i) \langle \mathcal{L} (\ket{u_j}\bra{u_j}),\ket{u_i}\bra{u_i}\rangle  = \langle \mathcal{L} (\ket{u_j}\bra{u_j}),\ket{u_i}\bra{u_i}\rangle_\rho = \langle \mathcal{L} (\ket{u_i}\bra{u_i}),\ket{u_j}\bra{u_j}\rangle_\rho = \pi(u_j) \langle \mathcal{L} (\ket{u_i}\bra{u_i}),\ket{u_j}\bra{u_j}\rangle.  
\end{equation}
Since $u_r$ is an eigenvector of $ \rho$ with corresponding eigenvalue $\rho(u_r) = \bra{u_r} \rho \ket{u_r} =\pi(u_r),$ we obtain:
\begin{align*}
   \langle  \ket{u_j}\bra{u_k}, \ket{u_r}\bra{u_r} \rangle_\rho &= \trace{\ket{u_j}\bra{u_k} \rho^{1/2} \ket{u_r}\bra{u_r}  \rho^{1/2} } = \trace{\ket{u_j}\bra{u_k}   (\rho(u_r)^{1/2}\ket{u_r}\bra{u_r}\rho(u_r)^{1/2})   } \\
   &=  \rho(u_r) \trace{\ket{u_j}\bra{u_k}  \ket{u_r}\bra{u_r}   }  =   \rho(u_r)  \braket{u_k}{u_r}  \braket{u_r}{u_j} = \pi(u_r) \delta_{k=r} \delta_{j=r}
\end{align*}
Substituting into \cref{eq:relate to classical spectral gap helper} gives
\begin{align*}
    \mathcal{E}_{\mathcal{L}} (f) &=  -\sum_{i, j}f_i f_j \pi(u_j) \langle \mathcal{L} (\ket{u_i}\bra{u_i}),\ket{u_j}\bra{u_j}\rangle \\
    &= \sum_{i} f_i^2  \pi(u_i) \sum_{j:j\neq i} P_{\mathcal{L},U}[u_i\to u_j] -\sum_{i\neq j} f_i f_j \pi(u_i) P_{\mathcal{L},U}[u_i\to u_j]\\
    &=\frac{1}{2}\sum_{i\neq j}   (f_i^2 + f_j^2 - 2f_i f_j) \pi(u_i) P_{\mathcal{L},U}[u_i\to u_j]\\
    &= \frac{1}{2}\sum_{i\neq j} (f_i-f_j)^2 \pi(u_i)P_{\mathcal{L},U}[u_i \to u_j] = \mathcal{E}_{\mathcal{L}, U,\mathrm{cl}} (F)
\end{align*}
where in the second equality we use \cref{eq:return probability for MC} and in the third we use the reversibility of the classical Markov generator. 
Next, 
\begin{align*}
    \Var_\rho(f) &= \trace{\rho^{1/2} f \rho^{1/2} f^\dag} - \trace{\rho f}^2 = \sum_{i} f_i^2 \bra{u_i} \rho \ket{u_i}  - \left(\sum_{i} f_i\bra{u_i} \rho \ket{u_i}\right)^2\\
    &= \sum_i F(u_i)^2 \pi(u_i) - \left(\sum_{i} F(u_i)  \pi(u_i)\right)^2\\
    &= \Var_\pi(F). \tag*{\qedhere}
\end{align*}

\end{proof}

\begin{proof}[Proof of \cref{thm:quantum vs classical spectral gap comparison}]
    Let $\hat{U}$ be an orthonormal eigenbasis of $H$ satisfying $ \hat{U}=\arg\min_{U} \lambda_{\mathcal{L},U,\mathrm{cl}}$. Then, by \cref{prop:classical chain V0}, 
    \[\lambda_{\mathcal{L},0} = \lambda_{\mathcal{L},0}^{\mathrm{H}} = \lambda_{\mathcal{L},\hat{U},\mathrm{cl}}, \]
    and by \cref{thm:general comparison theorem dependent only on length of arithmetic sequence} 
    \[ \lambda_{\mathcal{L},\hat{U},\mathrm{cl}} \geq \lambda_{\mathcal{L}} \geq \frac{1}{2D} \lambda_{\mathcal{L},\hat{U},\mathrm{cl}}.\qedhere\]
\end{proof}

To prove \cref{cor:cheeger}, we use a generalized version of Cheeger's inequality that applies to general Markov generators. Although the proof is essentially the same as in the standard case, we include it in \cref{sec:missing proofs} for completeness.
\begin{proposition} \label{prop:generalized cheeger}
Consider a classical Markov generator on a state space \( X \), with transition rates \( P[x \to y] \geq 0 \) for all \( x \neq y \in X \). Assume that the generator is \emph{reversible} with respect to a probability distribution \( \pi \) on \( X \)\footnote{\( \pi \) is not assumed to be the unique stationary distribution of the process.}; that is,
\[
\forall x \neq y:\quad \pi(x) P[x \to y] = \pi(y) P[y \to x].
\]
Recall that the Dirichlet form is given by 
\( \mathcal{E}(f) = \frac{1}{2} \sum_{x \neq y} (f(x) - f(y))^2 \pi(x) P[x \to y] \), 
and the spectral gap is 
\( \lambda = \inf_{f : X \to \mathbb{R}} \frac{\mathcal{E}(f)}{\operatorname{Var}_\pi(f)} \).

Define the \emph{bottleneck ratio} \( \Phi := \min_{S \subseteq X,\ \pi(S) \leq 1/2} \frac{Q(S, S^c)}{\pi(S)} \), where \( Q(S, S^c) = \sum_{x \in S,\ y \notin S} \pi(x) P[x \to y] \).

Let \( \Lambda := \sup_{x \in X} \sum_{y \neq x} P[x \to y] \) denote the \emph{maximum total transition rate} out of any state.
Then the following inequality holds:
\[
\Phi^2 \leq 2 \Lambda \lambda.
\]
 \end{proposition}
\begin{proof}[Proof of \cref{cor:cheeger}]

    By \cref{thm:quantum vs classical spectral gap comparison}, for a certain orthonormal eigenbasis $\hat{U} = \set{u_i}_i$ of $H,$
    \begin{equation}\label{eq:cheeger proof relating quantum and classical}
      \lambda_{\mathcal{L},\hat{U},\mathrm{cl}} \leq 2D  \lambda_{\mathcal{L}}.   
    \end{equation}

   We adopt the same notation as in \cref{def:classical Markov generator} and \cref{prop:classical chain V0}.

    Let $ \Lambda := \max_{u_i\in \hat{U}} \sum_{u_j:j\neq i} P_{\mathcal{L}, \hat{U}}[u_i\to u_j] . $ We prove that
\begin{equation}\label{eq:cheeger proof escaper rate}
  \Lambda\leq M  = \|G\|_\infty\cdot \|\sum_{S\in \mathcal{S}} S  S^\dag\|.  
\end{equation}
Indeed,
\begin{align*}
\sum_{u_j:j\neq i}P_{\mathcal{L}, \hat{U}}[u_i\to u_j] &= \sum_{u_j:j\neq i}\left( G(E_j -E_i) \sum_{S\in \mathcal{S}} \bra{u_i} S \ket{u_j}\bra{u_j} S^\dag \ket{u_i}\right) \\
&\leq \|G\|_\infty  \sum_{u_j}   \sum_{S\in \mathcal{S}} \bra{u_i} S \ket{u_j}\bra{u_j} S^\dag \ket{u_i}\\
&= \|G\|_\infty  \bra{u_i}  \left(\sum_{S\in \mathcal{S}} S (\sum_{u_j}\ket{u_j}\bra{u_j} ) S^\dag\right) \ket{u_i}\\
&= \|G\|_\infty \bra{u_i} (\sum_{S\in \mathcal{S}} S  S^\dag) \ket{u_i}\\
&\leq \|G\|_\infty\cdot \|\sum_{S\in \mathcal{S}} S  S^\dag\|
\end{align*}
As argued in \cref{prop:classical chain V0}, the classical Markov generator $P_{\mathcal{L},\hat{U}}$ is reversible wrt $\pi.$
    By Cheeger's inequality (see \cref{prop:generalized cheeger}), there exists a set $S\subseteq \hat{U}$ such that:
    \begin{equation}\label{eq:cheeger proof bottleneck S}
         \pi(S)\leq 1/2 \quad \text{and} \quad \sum_{u_i \in S, u_j\not\in S} \pi(u_i) P_{\mathcal{L}, \hat{U}}[u_i\to u_j] \leq \pi(S) \sqrt{2\Lambda  \lambda_{\mathcal{L},\hat{U},\mathrm{cl}} } .
    \end{equation}
     Let $F :\hat{U} \to \R$ be defined by $ F(u_i)= \textbf{1} [u_i \in S]$  and let $f = \sum F(u_i) \ket{u_i}\bra{u_i}.$ Observe that $f$ is a Hermitian projector that commutes with $H$ i.e. $f=f^\dag,$ $[f,H]=0$ and $f^2 = f.$
     By the same argument as in the proof of \cref{prop:classical chain V0} and a direct calculation, we have:
    \[\mathcal{E}_\mathcal{L} (f) = \mathcal{E}_{\mathcal{L}, U,\mathrm{cl}} (F) = \sum_{u_i \in S, u_j\not\in S} \pi(u_i) P_{\mathcal{L}, \hat{U}}[u_i\to u_j] \quad \text{and} \quad  \Var_\rho (f) = \Var_\pi (F) = (1-\pi(S)) \pi(S).\]
   Combining \cref{eq:cheeger proof relating quantum and classical,eq:cheeger proof bottleneck S,eq:cheeger proof escaper rate}, we get 
   \[ \mathcal{E}_\mathcal{L} (f) \leq \sqrt{2 \Lambda \lambda_{\mathcal{L},\hat{U},\mathrm{cl}} } \cdot \pi(S) \leq  2 \sqrt{2 M \lambda_{\mathcal{L},\hat{U},\mathrm{cl}} } \cdot \pi(S) (1-\pi(S)) = 2 \sqrt{2 M \lambda_{\mathcal{L},\hat{U},\mathrm{cl}} } \cdot\Var_\rho (f) \leq  4\sqrt{DM \lambda_{\mathcal{L}}} \cdot \Var_\rho (f)  . \]
\end{proof}

%% file: absence_of_long_arithmetic_progression.tex
\section{Absence of long proper arithmetic progressions}\label{sec:no AP in perturbed Hamiltonian}

In this section, we prove \cref{thm:generic transverse field}, i.e., perturbing any fixed Hamiltonian by a generic external field results in a new Hamiltonian whose spectrum contains no 3-arithmetic progressions or repeated eigenvalues. We can also show similar results for other types of perturbations. For example, we show that perturbing any fixed Hamiltonian by a generic degree-$k$ perturbation yields a new Hamiltonian whose spectrum contains no 3-arithmetic progression.

We recall the definition of the Pauli matrices:
\begin{align*}
X &:= 
\begin{bmatrix}
0 & 1 \\
1 & 0
\end{bmatrix}, \quad
Y := 
\begin{bmatrix}
0 & -i \\
i & 0
\end{bmatrix}, \quad
Z := 
\begin{bmatrix}
1 & 0 \\
0 & -1
\end{bmatrix}, \quad I:=\begin{bmatrix}
1 & 0 \\
0 & 1
\end{bmatrix}.
\end{align*}
Let \( P \in \mathbb{C}^{2 \times 2} \) be any single-qubit operator (such as one of the Pauli matrices).  
Then the operator \( P_i \), acting on the \( i \)th qubit of an \( n \)-qubit system, is defined as:
\[
P_i := I^{\otimes (i-1)} \otimes P \otimes I^{\otimes (n - i)}.
\]
\begin{proposition}\label{thm:generic higher order term}
    For any Hamiltonian $H_0$ and $k\in\N$,  the spectrum of the Hamiltonian $H := H_0 + \sum_{S\in \binom{[n]}{k}} A_{S} \prod_{i\in S} X_i$ almost surely (with respect to the randomness of $A = (A_{S})_{S \in \binom{[n]}{k}}\in \R^{ \binom{[n]}{k}}$) does not have any 3-arithmetic progressions, in the sense that the set of $A\in \R^{ \binom{[n]}{k}}$ such that $ H$ has a $3$-arithmetic progressions has Lebesgue measure $0$.
\end{proposition}
For quadratic perturbation, we can also restrict the domain of the interactions to be supported on an arbitrary connected graph.
\begin{proposition}\label{thm:generic quadratic term}
    For any given Hamiltonian $H_0$ and (simple) connected graph $G = G([n], E),$ the spectrum of the Hamiltonian $H := H_0 + \sum_{\{i, j\}\in E(G)} A_{ij} X_i X_j$ almost surely (with respect to the randomness of $A = (A_{ij})_{\{i,j\} \in E(G)}\in \R^{E(G)}$) does not have any 3-arithmetic progressions, in the sense that the set of $A\in \R^{E(G)}$ such that $ H$ has a $3$-arithmetic progressions has Lebesgue measure $0$.
\end{proposition}
As corollaries, we obtain the following results regarding the spectrum of the transverse field Ising model.
\begin{corollary}
    For any simple connected graph $G = G([n], E)$ and fixed external field $H,$ the spectrum of the Hamiltonian $H = \sum_{\{i, j\}\in E(G)} A_{ij} Z_i Z_j + h \sum_i X_i$ almost surely (with respect to the randomness of $A = (A_{ij})_{\{i,j\} \in E(G)}$) does not have any 3-arithmetic progressions, in the sense that the set of $A\in \R^{E(G)}$ such that the spectrum of $ H$ has a $3$-arithmetic progressions has Lebesgue measure $0$.
\end{corollary}
\begin{corollary}
    For any interaction matrix $A=(A_{ij})_{i,j},$ the spectrum of the Hamiltonian $ H = \sum_{i,j} A_{ij} Z_i Z_j + \sum_i h_i X_i$ almost surely (with respect to the randomness of $\boldsymbol{h} = (h_i)_{i=1}^n\in \R^n$) does not have any 3-arithmetic progressions or repeated eigenvalues, the sense that the set of $\boldsymbol{h}\in \R^n$ such that the spectrum of $H$ has any 3-arithmetic progressions or repeated eigenvalues has Lebesgue measure $0.$
\end{corollary}
The above shows that the spectrum of Hamiltonians with \emph{varying} interactions or external fields generically does not contain a long arithmetic progression. We can also to show that this property holds for some Hamiltonians with uniform interaction and external field.

\begin{proposition}\label{thm:1D XY+Z}
     Let $H$ be the 1D XY+Z Hamiltonian with periodic boundary condition, i.e.,
     \begin{equation}\label{eq:1D XY+Z periodic boundary condition}
         H  = J \sum_{i=1}^n X_i Y_{i+1} +  h \sum_{i=1}^n Z_i \quad \text{where}\quad Y_{n+1}= Y_1. \nonumber
     \end{equation}
     For an odd prime $n,$ the spectrum of $H$ almost surely (with respect to the randomness of $J$ and $h$) does not contain any 3-arithmetic progressions or repeated eigenvalues, in the sense that the set of $ (J,h)\in \R^2$ such that the spectrum of $ H$ has 3-arithmetic progression or repeated eigenvalues or repeated eigenvalues has Lebesgue measure $0$.\footnote{For odd primes $n$, $ H$ almost surely does not have repeated eigenvalues; this was proved in \cite[Lemma 1]{Keating_2015}.}
\end{proposition}
We can also show that perturbing any fixed Hamiltonian $H_0$ by the 1D XY+Z Hamiltonian yields a Hamiltonian whose spectrum does not contain any 3-arithmetic progression or repeated eigenvalues. This class includes several translationally invariant 1D models studied in the literature (see \cite{Keating_2015}).
\begin{proposition} \label{thm:perturbation by 1D XY+Z}
    For any given Hamiltonian $H_0,$ the spectrum of the Hamiltonian
    \[  H  :=  H_0 + J \sum_{i=1}^n X_i Y_{i+1} +  h \sum_{i=1}^n Z_i \quad \text{where}\quad Y_{n+1}= Y_1\]
    almost surely (with respect to the randomness of $J$ and $h$) 
    does not contain any 3-arithmetic progression or repeated eigenvalues, in the sense that the set of $ (J,h)\in \R^2$ such that the spectrum $ H$ has any 3-arithmetic progressions or repeated eigenvalues has Lebesgue measure $0$.
\end{proposition}

We first prove \cref{thm:generic transverse field}, which will be a model for the other proofs in this section. 

\begin{proof}[Proof of \cref{thm:generic transverse field}]
We use a similar strategy as in \cite{HuangHarrow25}.
Let the eigenvalues of $H$ be $\lambda_1 , \lambda_2, \dots, \lambda_N.$
Let \[F(h_1, \dots, h_n) := \prod_{i,j,k \text{ distinct}} (\lambda_i + \lambda_j - 2\lambda_k)\] and \[G (h_1, \dots, h_n) := \prod_{i,j: i\neq j} (\lambda_i-\lambda_j).\] 
Note that the spectrum of $H$ has a $3$-term arithmetic progression if and only if $F = 0, $ and similarly, the spectrum $H$ has eigenvalues with multiplicities if and only if $G=0.$ Since $F$ and $G$ are symmetric under permutation of the eigenvalues of $H$, they are symmetric polynomials in the eigenvalues $\lambda_1, \dots, \lambda_N,$ and thus by the fundamental theorem of symmetric polynomials, can be written as a polynomial in terms of $ \sum_{i=1}^N \lambda_i^k = \tr(H^k)$ for $k\in [N].$ Since $\tr(H^k) $ are polynomials in $h_1, \dots, h_n$ (treating the entries of $H_0$ as fixed constants), $F$ and $G$ are polynomials in $h_1, \dots, h_n.$ Since the zeros of a multivariate polynomial are of measure zero unless the polynomial is identically zero, it suffices to find a particular $\boldsymbol{h} = (h_i)_{i=1}^n$ such that $F(\boldsymbol{h})$ and $G(\boldsymbol{h}) $ are nonzero. Pick $R> 0$, and let $\tilde{h}_i := R^{-1} h_i$ and $ \tilde{H}_{\tilde{\boldsymbol{h}}} := R^{-1} H_0 + \sum_{i=1}^n \tilde{h}_i P_i$. Then, the eigenvalues of $\tilde{H}_{\tilde{\boldsymbol{h}}}$ are $\tilde{\lambda}_i := R^{-1} \lambda_i $ for $i\in [N].$ We can rewrite
\begin{align*}
    F(\boldsymbol{h}) &= R^{N(N-1)(N-2)} \prod_{i,j,k \text{ distinct}} (\tilde{\lambda}_i+ \tilde{\lambda}_j - 2 \tilde{\lambda}_k),\\
    G(\boldsymbol{h}) &= R^{N(N-1)} \prod_{i,j: i\neq j} (\tilde{\lambda}_i- \tilde{\lambda}_j).
\end{align*}
Let the eigenvalues of the Hermitian matrix $P\in \C^{2\times 2}$ be $ c_1, c_2 \in \R,$ where $ c_1 \neq c_2,$ since $P$ is not a multiple of the identity.

Note that the eigenvalues of $\sum_{i=1}^n \tilde{h}_i P_i $ are $ \lambda_j^* = \frac{c_1-c_2}{2} \sum_{i=1}^n (2x_i -1) \tilde{h}_i + C$ where $x_1 \cdots x_n \in \set{0,1}^n$ is the binary representation of $j,$ and $C = \frac{\lambda_1+\lambda_2}{2}\sum_i\tilde{h}_i.$ 
Keeping $\tilde{\boldsymbol{h}}$ fixed while taking $R \to +\infty,$
$ \prod_{i,j,k \text{ distinct}} (\tilde{\lambda}_i+ \tilde{\lambda}_j - 2 \tilde{\lambda}_k)$ (and $ \prod_{i,j: i\neq j} (\tilde{\lambda}_i- \tilde{\lambda}_j) $, respectively) converges to $ \prod_{i,j,k \text{ distinct}} (\lambda^*_i+ \lambda^*_j - 2 \lambda^*_k)$ ($ \prod_{i,j: i\neq j} (\lambda^*_i- \lambda^*_j )$ respectively), so it suffices to exhibit a particular $\tilde{\boldsymbol{h}}$ such that $ \prod_{i,j,k \text{ distinct}} (\lambda^*_i+ \lambda^*_j - 2 \lambda^*_k)\neq 0 $ and $ \prod_{i,j: i\neq j} (\lambda^*_i- \lambda^*_j ) \neq 0,$ which is equivalent to $\sum_{i=1}^n \tilde{h}_i P_i$ not having any $3$-APs or repeated eigenvalues in its spectrum. The existence of such an $\tilde{\boldsymbol{h}}$ is guaranteed by the following proposition.
 \end{proof}
\begin{proposition}\label{prop:generic transverse field helper}
 If $ \tilde{h}_1, \dots, \tilde{h}_n$ are linearly independent over $\Z,$ then $\sum_{i=1}^n \tilde{h}_i P_i$ does not have any $3$-APs or repeated eigenvalues in its spectrum.
\end{proposition}
\begin{proof}
The eigenvalues of $\sum_{i=1}^n \tilde{h}_i P_i$ are $\frac{c_1-c_2}{2} \sum_{i=1}^n z_i \tilde{h}_i  +C $ for $z\in \set{\pm 1}^n$ and $C = \frac{c_1+c_2}{2}\sum_i\tilde{h}_i,$ where $c_1 ,c_2\in \R$ are the eigenvalues of $P, $ and note that $c_1 \neq c_2$ since $P$ is not a multiple of the identity.

 Suppose for contradiction that $ \sum_{i=1}^n P_i \tilde{h}_i $ has repeated eigenvalues. Then there exist distinct $v,z\in \set{\pm 1}^n$ such that 
\[ 0 = \sum_{i=1}^n v_i \tilde{h}_i - \sum_{i=1}^n z_i \tilde{h}_i = \sum_{i=1}^n (v_i - z_i) \tilde{h}_i.\]
Since $ v, z$ in are distinct vectors in $\set{\pm 1}^n$, $ (v_i-z_i)_{i=1}^n\in \set{0,\pm 2}^n \subseteq \Z^n$ and is not the zero vector, so $ (\tilde{h}_i)_{i=1}^n$ are not linearly independent over $\Z,$ a contradiction.

Suppose for contradiction that there exists a $3$-AP, i.e., there exist (pairwise) distinct $v, w, z \in \set{\pm 1}^n$ such that 
\[ 0 = \sum_{i=1}^n v_i \tilde{h}_i + \sum_{i=1}^n w_i \tilde{h}_i - 2 \sum_{i=1}^n z_i \tilde{h}_i = \sum_{i=1}^n (v_i + w_i -2 z_i) \tilde{h}_i.\]
Since $v$ and $w$ are distinct, there exists $i$ such that $ v_i \neq w_i,$ and thus $v_i + w_i -2z_i= 1 + (-1) -2 z_i= -2z_i \neq 0.$ So $(v_i+w_i-2z_i)_{i=1}^n\in \set{0,\pm 2,\pm 4}^n \subseteq \Z^n$ and is not the zero vector, contradicting that $ (\tilde{h}_i)_{i=1}^n$ are linearly independent over $\Z.$ 
\end{proof}

\begin{proof}[Proof of \cref{thm:generic quadratic term}]
 By the same argument as in proof of \cref{thm:generic transverse field}, the statement reduces to the following proposition.
\end{proof}
\begin{proposition}
 If $G$ is connected and $ (A_{ij})_{\{i,j\} \in E(G)}$ are linearly independent over $\Z,$ then the spectrum of $ \sum_{\{i, j\}\in E(G)} A_{ij} X_i X_j $ does not have any 3-term arithmetic progression.
 \end{proposition}
 \begin{proof}
 
 Note that $ \sum_{\{i, j\}\in E(G)} A_{ij} X_i X_j$ is equivalent to $ \sum_{\{i, j\}\in E(G)} A_{ij} Z_i Z_j$ by a unitary transformation, so we only need to work with $ \sum_{\{i, j\}\in E(G)} A_{ij} Z_i Z_j.$ 
 
 The computational basis forms an eigenbasis of $ \sum_{\{i, j\}\in E(G)} A_{ij} Z_i Z_j$ where the eigenvalue corresponding to $ y\in \set{0,1}^n$ is $\sum_{\{i, j\}\in E(G)} A_{ij} (2y_i-1) (2 y_j-1) .$ The eigenvalues of $ \sum_{\{i, j\}\in E(G)} A_{ij} Z_i Z_j $ are thus $ \sum_{\{i, j\}\in E(G)} A_{ij} z_i z_j $ for $z\in \set{\pm 1}^n.$ Suppose for contradiction that there exists a 3-AP, i.e.\ there exist pairwise distinct $ v,w,z \in\set{\pm 1}^n$ such that 
 \[ 0 = \sum_{\{i,j\}\in E(G)} v_i v_j A_{ij} + \sum_{\{i,j\}\in E(G)} w_i w_j A_{ij} - 2 \sum_{\{i,j\}\in E(G)} z_i z_j A_{ij}= \sum_{\{i,j\}\in E(G)} (v_i v_j + w_i w_j - 2 z_i z_j) A_{ij}.\]
 It is easy to see that $ v_i v_j + w_i w_j - 2z_i z_j \in \{0,\pm 2, \pm 4\},$ so for linearly independent $(A_{ij}),$ the above is true iff $v_i v_j + w_i w_j - 2 z_i z_j = 0$ for all $\{i,j\} \in E(G). $ This implies that $v_i v_j =w_i w_j = z_i 
 z_j$ for all $\{i,j\} \in E(G).$ 
 
 Indeed, suppose there exists $\{i,j\}\in E(G)$ such that $ v_i v_j \neq w_i w_j $ then since $v_i v_j, w_i w_j \in \{\pm 1\},$ $v_i v_j + w_i w_j -2 z_i z_j= 0 -2z_i z_j \neq 0 ,$ a contradiction. Thus for all $\{i,j\}\in E(G):$ $v_i v_j = w_i w_j ,$ which then implies $v_i v_j = w_i w_j = z_i z_j$ since $v_i v_j + w_iw_j - 2z_i z_j = 0. $

 Fix an arbitrary $i\in[n].$ By the pigeonhole principle, since $v_i, w_i , z_i \in \{\pm 1\},$ either $ v_i =w_i $ or $v_i =z_i.$ Wlog assume $v_i = w_i.$ Then for any $j $ such that $ \{i,j\} \in E(G),$ $v_j = w_j$ (since $v_i v_j = w_i w_j$). By this same argument, any $ j$ such that there is a path in $G$ from $i$ to $j$ has $v_j = w_j,$ and this includes all vertices in $G$ since $G$ is connected. Thus $v$ is identical to $w,$ contradicting that $v$ and $w$ are distinct vectors in $\{\pm 1\}^n.$
 \end{proof}

\begin{proof}[Proof of \cref{thm:generic higher order term}]
 By the same argument as in proof of \cref{thm:generic transverse field}, the statement reduces to the following proposition.
\end{proof}
 \begin{proposition}
 If $ (A_{S})_{S \in \binom{[n]}{k} }$ are linearly independent over $\Z,$ then the spectrum of $ \sum_{S \in \binom{[n]}{k} } A_{S} \prod_{i \in S}X_i $ does not have any 3-term arithmetic progression.
 \end{proposition}
 \begin{proof}
 Note that \(\sum_{S \in \binom{[n]}{k}} A_S \prod_{i \in S} X_i\) is unitarily equivalent to \(\sum_{S \in \binom{[n]}{k}} A_S \prod_{i \in S} Z_i\), so it suffices to work with the latter. The computational basis forms an eigenbasis for \(\sum_{S \in \binom{[n]}{k}} A_S \prod_{i \in S} Z_i\), with eigenvalues indexed by \(y \in \{0,1\}^n\) given by \(\sum_{S \in \binom{[n]}{k}} A_S \prod_{i \in S} (2 y_i - 1)\). Equivalently, the eigenvalues are \(\sum_{S \in \binom{[n]}{k}} A_S \prod_{i \in S} z_i\) for \(z \in \{\pm 1\}^n\).

Suppose for contradiction that there exists a 3-AP,i.e. then by the same argument as in the proof of \cref{thm:generic quadratic term}, there exists pairwise distinct $ v,w,z \in\set{\pm 1}^n$ such that for all $ S\in \binom{[n]}{k}$, $\prod_{i\in S} v_i = \prod_{i\in S} w_i =\prod_{i\in S} z_i.$
We prove the following claim by induction on $k\geq 1,$ and concludes that the above implies either $v=w $ or $v=z,$ contradicting that $v,w,z$ are pairwise distinct. 
 \begin{claim}
 Let $v,w\in \{\pm 1\}^n$ such that for all $S\in\binom{[n]}{k}$, $\prod_{i\in S} v_i = \prod_{i\in S} w_i.$ If $k$ is odd then $v_i =w_i$ for all $i\in [n].$ If $k$ is even then either $v=w$ or $v = -w.$
 \end{claim}

 The base case $k = 1$ is trivially true. Consider $k\geq 2$, and suppose we have proved the claim for all $k'\leq k-1.$ There are two cases, $k$ is odd and $k$ is even.
 
 Suppose $k$ is odd. 
 If $ v_i\neq w_i$ for all $i\in [n],$ then since $v_i, w_i\in\{\pm1\},$ it must be that $v_i = -w_i$ for all $i \in [n],$ but this implies $ \prod_{i\in S} v_i = (-1)^k\prod_{i\in S} w_i = -\prod_{i\in S} w_i \neq \prod_{i\in S} w_i,$ a contradiction. Hence, there exists $ i^*$ such that $ v_{i^*} = w_{i^*}.$ This means for any $T\subseteq \binom{[n]\setminus \{i^*\} }{k-1},$ by factoring out $v_{i^*} = w_{i^*}$ from $\prod_{i\in T\cup \{i^*\}} v_{i} =\prod_{i\in T\cup \{i^*\}} w_{i}, $ we have $ \prod_{i\in T} v_i = \prod_{i\in T} w_i.$ Apply the induction hypothesis for $k-1$, and note that $k-1$ is even, we have that either $ v_i = w_i$ for all $i\in [n]$ or $v_i =-w_i$ for all $i\in [n]\setminus \{i^*\}.$ In the second case, for any $S\subseteq \binom{[n]\setminus \{i^*\}}{k},$ $ \prod_{i\in S} v_i =- \prod_{i\in S} w_i \neq \prod_{i\in S} w_i$ (a contradiction). So the first case i.e. $v_i=w_i$ for all $i\in [n]$ holds.

 Suppose $k$ is even. Fix an arbitrary index $i^*,$ then there are two cases: either $ v_{i^*} = w_{i^*}$ or $v_{i^*} = - w_{i^*}.$ If $v_{i^*} = w_{i^*}$ then by the same argument as above, $ \prod_{i\in T} v_i =\prod_{i\in T} w_i$ for any $T\in \binom{[n]\setminus \{i^*\}}{k-1},$ which, by the inductive hypothesis, implies that $v_i = w_i$ for all $i\in [n].$ For the second case, $ v_{i^*} = -w_{i^*}.$ Let $\tilde{v}_i = -v_i$ for all $i,$ then $\tilde{v}_{i^*}=w_{i^*}$ and $\prod_{i\in T} \tilde{v}_i = \prod_{i\in T} w_i$ for any $T\in \binom{[n]\setminus \{i^*\}}{k-1},$ thus $ \tilde{v}_i =w_i$ for all $i\in [n]$ by the induction hypothesis.
 \end{proof}

\begin{proof}[Proof of \cref{thm:1D XY+Z}]
 We normalize the Hamiltonian so that $h =1 .$
 By \cite{Keating_2015}, the eigenvalues of $ H$ are $\lambda_z (J)= \sum_{k=1}^n z_k m_k(J) $ where $m_k(J) = J \mu_k- \sqrt{J^2 \mu_k^2 +1}$ and $ \mu_k = \sin\left(\frac{2\pi k}{n}\right),$ for $z\in \{\pm 1\}^n.$ Let $n$ be an odd prime.



The set of \(J \in \mathbb{R}\) for which \(H\) contains a 3-term arithmetic progression (3-AP) is
\[
\bigcup_{\substack{v,w,z \in \{\pm 1\}^n \\ \text{distinct}}} \{ J \in \mathbb{R} \mid \lambda_v(J) + \lambda_w(J) - 2 \lambda_z(J) = 0 \}.
\]
Thus, it suffices to show that for each distinct triple \(v,w,z \in \{\pm 1\}^n\), the set
\[
\{ J \in \mathbb{R} \mid \lambda_v(J) + \lambda_w(J) - 2 \lambda_z(J) = 0 \}
\]
has measure zero. Since \(\lambda_v(J) + \lambda_w(J) - 2 \lambda_z(J)\) is a real-analytic function of \(J\), its zero set has measure zero unless the function is identically zero.

 Suppose that $ \lambda_{v} (J) +\lambda_z(J) -2\lambda_w(J)$ is identically $0,$ then for $d_k = \frac{v_k+w_k-z_k}{2}\in \set{0,\pm 1,\pm 2},$ we have: 
 \begin{equation}\label{eq:1D uniform external field main}    \forall J: \sum_{k=1}^n d_k m_k(J) = 0
 \end{equation}
 We will show the following claim:
 \begin{claim}
     If \cref{eq:1D uniform external field main} holds for some $d_1, \cdots, d_n \in \Z$ then $d_1 = \cdots = d_n =0.$
 \end{claim}
Given this claim, we conclude that $ 0=2d_k =  v_k+w_k-z_k\forall k\in [n].$ But, since $v,w$ are distinct, there exists $k\in [n]$ such that $ v_k \neq w_k,$ in which case, $ 2 d_k = v_k+w_k -2z_k = -2z_k \neq 0,$ a contradiction.

Now, we verify the claim. 

 First, observe that for odd prime $n$, $\mu_1,\cdots, \mu_{(n-1)/2}$ are linearly independent over $\Z.$ Suppose for contradiction that there exists $ a_1, \cdots, a_{(n-1)/2}\in \Z$ such that $0 = \sum_{j=1 }^{(n-1)/2}  a_j \mu_j.$ Since each $\mu_j$ can be rewritten as $\mu_j = \frac{\omega^j -\omega^{-j}}{2i} $ where $\omega = e^{2i\pi/n},$ this would imply that $\set{\omega_j}_{j=1}^{n-1}$ is linearly dependent over $\Z$, a contradiction (see \cite[Lemma 2.11]{lenstra1979vanishing}).

 For small $J,$ Taylor series around $J =0$ for \cref{eq:1D uniform external field main} gives
 \[0=- 2\sum_k d_k + 2J \sum_{k} \mu_k d_k -J^2 \sum_{k} \mu_k^2 d_k + O(J^3). \]
This implies:
\[ 0 =\sum_{k=1}^n \mu_k d_k \]
Using the facts that
\[\mu_n =\sin(2\pi) =0\quad \text{and} \quad \forall k\in [n-1]: \mu_k =  \sin\left(\frac{2\pi k}{n}\right) = -  \sin\left(\frac{2\pi (n-k)}{n}\right)=-\mu_{n-k},\]
we can rewrite
\[  0 =\sum_{k=1}^n \mu_k d_k = \sum_{k=1}^{(n-1)/2} \mu_k (d_k - d_{n-k}) \]
Since the set $\set{\mu_k}_{k=1}^{(n-1)/2}$ is linearly independent over $\Z,$ this implies \[d_k = d_{n-k}=0\forall k\in\set{1, \cdots, (n-1)/2}.\] Substituting this into \cref{eq:1D uniform external field main}, and using $m_n(J) =  -1$ and $\mu_k = -\mu_{n-k},$ we get: 
\begin{align*}
    0 &= \sum_{k=1}^n d_k m_k(J) \\
    &=-d_n+\sum_{k=1}^{n-1} J\mu_k d_k - \sum_{k=1}^{n-1} d_k \sqrt{J^2 \mu_k^2 + 1}    \\
    &= -d_n +J \sum_{k=1}^{(n-1)/2} d_k (\mu_k + \mu_{n-k})- 2\sum_{k=1}^{(n-1)/2} d_k \sqrt{J^2 \mu_k^2 + 1} \\
    &= -d_n -  2\sum_{k=1}^{(n-1)/2} d_k \sqrt{J^2 \mu_k^2 + 1} 
\end{align*}
This equation holds for all $J.$ Dividing both sides by $ J$ yields:
\[ 0 = -\frac{d_n}{J} - 2\sum_{k=1}^{(n-1)/2} d_k \sqrt{\mu_k^2 + (1/J)^2} \]
Taking the limit $J \to +\infty,$ and using the fact that $\mu_k =  \sin\left(\frac{2\pi k}{n}\right) \geq 0\forall k \in \set{1, \cdots, (n-1)/2},  $ we obtain:
\[ 0 = -2 \sum_{k=1}^{(n-1)/2} d_k\mu_k  \]
Again, using the linear independence of $\set{\mu_k}_{k=1}^{(n-1)/2}$ over $\Z,$ we conclude $  d_k =0 \forall \in\set{1, \cdots, (n-1)/2}.$ Since $d_k = d_{n-k}\forall k\in\set{1, \cdots, (n-1)/2},$ this implies $d_k =0\forall k\in [n-1].$ Finally, substituting into \cref{eq:1D uniform external field main} and using $m_n(J) =-1$ yields $d_n = 0$, which finishes the claim.

Now, we prove that the set of $J$ s.t. $ H$ contains repeated eigenvalues has measure $0$\footnote{This was shown in \cite{Keating_2015}, but we include a proof for the sake of completeness}. By the same argument above, we only need to verify that for $ v, w \in \set{\pm 1}^n,$ if $  \lambda_{v} (J) -\lambda_w(J) $ is identically $0$ then $v=w.$ Indeed, letting $d'_k  = \frac{v_k -w_k}{2}\in \set{0, \pm 1}\subseteq \Z,$ if
\[  \lambda_v(J) - \lambda_w(J) =2\sum_{k=1}^n d'_k m_k(J) \]
is identically $0$ then the above claim implies $d'_k =0\forall k\in [n]$ and thus $v=w.$
\end{proof}

\begin{proof}[Proof of \cref{thm:perturbation by 1D XY+Z}]
    By the same argument as in proof of \cref{thm:generic transverse field}, the statement reduces to \cref{thm:1D XY+Z}.
\end{proof}

%% file: appendix.tex
\appendix

\section{Deferred proofs} \label{sec:missing proofs}
\begin{proof}[Proof of \cref{prop:simple comparison result when orbit has dimension 1}]
    Since $\mathcal{L}(\ket{u_i}\bra{u_j})\subseteq \text{span} (\ket{u_i}\bra{u_j}) ,$ we can write $ \mathcal{L}(\ket{u_i}\bra{u_j})= c_{i,j} \ket{u_i}\bra{u_j}  $ for some $c_{i,j}\in \C.$

We use the same notations as in \cref{def:classical Markov generator}.
    Since for $i\neq j,$ each $\text{span} (\ket{u_i}\bra{u_j})$ forms an invariant subspace under $\mathcal{L},$ we only need to lower bound the spectral gap in this subspace. Note that for any observable $f$ in this subspace, we can write $f = c\ket{u_i}\bra{u_j}$ with $ c\in \C,$ $ \mathcal{E}_\mathcal{L}(f) =-\sqrt{\pi(u_i)\pi(u_j)} \lvert c\rvert^2 c_{ij} $ and $ \Var(f) = \sqrt{\pi(u_i)\pi(u_j)} \lvert c\rvert^2 ,$ so if $\pi(i)\pi(j)=0$ then $ \lambda(f) =+\infty$ and we are done. From now on, assume $\pi(u_i),\pi(u_j)\neq 0.$ In this case, the spectral gap in this subspace is $ -c_{ij}.$
    
    Let $E_i$ be the eigenvalue corresponding to eigenvector $u_i.$
One can directly calculate (see also \cite[Theorem 11]{temme13}):
\begin{align*}
    -c_{ij} &= -\langle\ket{u_i}\bra{u_j} , \mathcal{L}(\ket{u_i}\bra{u_j})\rangle \\
    &= \sum_{S} -\left(G(0) \bra{u_i} S^\dag\ket{u_i} \bra{u_j} S\ket{u_j}- \frac{1}{2}(\sum_{k} G(E_k - E_j) \lvert\bra{u_j} S \ket{u_k}\rvert^2 + \sum_{k} G(E_k - E_i) \lvert\bra{u_i} S \ket{u_k}\rvert^2\right)  \\
    &= \frac{1}{2} G(0)\sum_S\lvert\bra{u_i} S \ket{u_i} - \bra{u_j} S \ket{u_j}\rvert^2   + \frac{1}{2}\left(\sum_{k\neq j} G(E_k - E_j)\sum_S \lvert\bra{u_j} S \ket{u_k}\rvert^2 + \sum_{k\neq i} G(E_k - E_i) \sum_S \lvert\bra{u_i} S \ket{u_k}\rvert^2\right) \\
    &= \frac{1}{2} G(0)\sum_S\lvert\bra{u_i} S \ket{u_i} - \bra{u_j} S \ket{u_j}\rvert^2  +\frac{1}{2}(\sum_{k\neq j} P_{\mathcal{L},U} [u_j \to u_k] +\sum_{k\neq i} P_{\mathcal{L},U} [u_i \to u_k] ) \\
    &\geq \frac{1}{2}(\sum_{k\neq j} P_{\mathcal{L},U} [u_j \to u_k] +\sum_{k\neq i} P_{\mathcal{L},U} [u_i \to u_k] )\\
    &\geq \frac{1}{2} \lambda_{\mathcal{L},U,\text{cl}},
\end{align*}
which finishes the claim.

 To prove the last inequality, consider the test function $F : U\to \R$ where $F (u_k) = \textbf{1}[k=i],$ and noting that the classical Markov generator is reversible (see proof of \cref{prop:classical chain V0}), then
 \[ \mathcal{E}_{\mathcal{L}, U, \text{cl}}(F) = \frac{1}{2}(\sum_{k\neq i}\pi(u_i) P_{\mathcal{L},U} (u_i \to u_k) + \sum_{k\neq i}\pi(u_k) P_{\mathcal{L},U} (u_k \to u_i)) = \pi(u_i) \sum_{k\neq i}P_{\mathcal{L},U} (u_i \to u_k) \]
 and
 \[\Var_\pi(F) = \pi(u_i)(1-\pi(u_i));\]
 thus
 \[ \sum_{k\neq i}P_{\mathcal{L},U} (u_i \to u_k) \geq (1-\pi(u_i))\lambda_{\mathcal{L},U,\text{cl}}.\]
 Similarly, 
 \[\sum_{k\neq j} P_{\mathcal{L},U} [u_j \to u_k]\geq (1-\pi(u_j))\lambda_{\mathcal{L},U,\text{cl}},\]
 so
 \[\sum_{k\neq j} P_{\mathcal{L},U} [u_j \to u_k] +\sum_{k\neq i} P_{\mathcal{L},U} [u_i \to u_k]  \geq (2-\pi(u_i)-\pi(u_j)) \lambda_{\mathcal{L},U,\text{cl}} \geq \lambda_{\mathcal{L},U,\text{cl}}.\]
    

     When $ H$ has simple spectrum and simple Bohr spectrum, each subspace $V_\omega$ has dimension $1,$ so the precondition holds and $H$ has a unique orthonormal eigenbasis, so $\lambda_{\mathcal{L}}\geq \frac{1}{2}\lambda_{\mathcal{L},\text{cl}}$. 
\end{proof}
\begin{proof}[Proof of \cref{prop:dirichlet form as sum of commutators}]
Note that $\mathcal{S}$ is self-adjoint, and $S^\dag(-\omega) = S(\omega)^\dag, $ so we can write
\begin{align}
&2 \langle \mathcal{L}_{\omega, S}(f) + \mathcal{L}_{-\omega, S^\dag}(f), f \rangle_\rho \notag \\
&\qquad= G(\omega) \left\langle S(\omega)^\dag [f, S(\omega)] + [S(\omega)^\dag, f] S(\omega), f \right\rangle_\rho \notag \\
&\qquad\qquad + G(-\omega) \left\langle S^\dag(-\omega)^\dag [f, S^\dag(-\omega)] + [S^\dag(-\omega)^\dag, f] S^\dag(-\omega), f \right\rangle_\rho \notag \\
&\qquad= G(\omega) \left( \tr\left( S(\omega)^\dag [f, S(\omega)] \rho^{1/2} f^\dag \rho^{1/2} \right) 
     + \tr\left( [S(\omega)^\dag, f] S(\omega) \rho^{1/2} f^\dag \rho^{1/2} \right) \right) \notag \\
&\qquad\qquad + G(-\omega) \left( \tr\left( S(\omega) [f, S(\omega)^\dag] \rho^{1/2} f^\dag \rho^{1/2} \right) 
     + \tr\left( [S(\omega), f] S(\omega)^\dag \rho^{1/2} f^\dag \rho^{1/2} \right) \right) \notag \\
&\qquad=_{\text{(1)}} G(\omega) e^{\beta \omega /2} \left( \tr\left( [f, S(\omega)] \rho^{1/2} f^\dag S(\omega)^\dag \rho^{1/2} \right) 
     + \tr\left( [S(\omega)^\dag, f] \rho^{1/2} S(\omega) f^\dag \rho^{1/2} \right) \right) \notag \\
&\qquad\qquad + G(-\omega) e^{-\beta \omega /2} \left( \tr\left( [f, S(\omega)^\dag] \rho^{1/2} f^\dag S(\omega) \rho^{1/2} \right)
     + \tr\left( [S(\omega), f] \rho^{1/2} S(\omega)^\dag f^\dag \rho^{1/2} \right) \right) \notag \\
&\qquad= \tilde{G}(\omega) \left( \tr\left( [f, S(\omega)] \rho^{1/2} f^\dag S(\omega)^\dag \rho^{1/2} \right) 
     + \tr\left( [S(\omega), f] \rho^{1/2} S(\omega)^\dag f^\dag \rho^{1/2} \right) \right) \notag \\
&\qquad\qquad + \tilde{G}(-\omega) \left( \tr\left( [S(\omega)^\dag, f] \rho^{1/2} S(\omega) f^\dag \rho^{1/2} \right)
     + \tr\left( [f, S(\omega)^\dag] \rho^{1/2} f^\dag S(\omega) \rho^{1/2} \right) \right) \notag \\
&\qquad=_{\text{(2)}} -\tilde{G}(\omega) \tr \left( [S(\omega), f] \rho^{1/2} [S(\omega), f]^\dag \rho^{1/2} \right) 
     - \tilde{G}(-\omega) \tr \left( [S(\omega)^\dag, f] \rho^{1/2} [S(\omega)^\dag, f]^\dag \rho^{1/2} \right) \notag \\
&\qquad=_{\text{(3)}} -\tilde{G}(\omega) \| [S(\omega), f] \|_\rho^2 
     - \tilde{G}(-\omega) \| [S^\dag(-\omega), f] \|_\rho^2,\nonumber
\end{align}
where (1) is due to applying \cref{prop:operator in V omega commutation with rho} to $S(\omega)\in V_\omega$ and $S(\omega)^\dag \in V_{-\omega}$, (2) is by rearranging terms and the fact that $\tilde{G}(\omega) = G(\omega) e^{\frac{\beta \omega }{2}} = \tilde{G}(-\omega),$ and (3) is by definition of $\| \cdot \|_\rho$ and $ S(\omega)^\dag = S^\dag (-\omega).$ Hence, we obtain the desired equality
\begin{align*}
\mathcal{E}_{\mathcal{L}}(f) &= -\left\langle \sum_{S\in \mathcal{S}, \omega} \mathcal{L}_{\omega, S} (f) , f \right\rangle_\rho \\
&= - \frac{1}{2} \sum_{S\in \mathcal{S}, \omega} \langle \mathcal{L}_{\omega, S}(f) + \mathcal{L}_{-\omega, S^\dag}(f), f \rangle_\rho
 \\
 &= \frac{1}{4} \sum_{S\in \mathcal{S}, \omega} ( \tilde{G}(\omega) \| [S(\omega), f] \|_\rho^2 
     + \tilde{G}(-\omega) \| [S^\dag(-\omega), f] \|_\rho^2)= \frac{1}{2} \sum_{S,\omega}  \tilde{G}(\omega) \| [S(\omega), f]\|_\rho^2. \qedhere    
\end{align*}

\end{proof}

\begin{proof}[Proof of \cref{prop:generalized cheeger}]
For $x\neq y,$ let $Q(x,y) = \pi(x) P[x\to y] .$
We adapt the standard proof of Cheeger's inequality from \cite{LP17}.
We first show the following claim:
\begin{claim}\label{claim:cheeger helper claim}
    Given a function $\psi:X \to \R_{\geq 0},$ order $X$ so that $\psi$ is non-increasing. If $\pi( \psi > 0)\leq 1/2$ then
\[
\E_\pi[\psi] \leq \Phi^{-1} \sum_{\substack{x, y \in X \\ x < y}} (\psi(x) - \psi(y)) Q(x,y)
\]
\end{claim}
\begin{proof}
    Let $S = \{x : \psi(x) > t\}$ with $t > 0$. 
    Since $ \pi(S) \leq \pi( \psi > 0)\leq 1/2,$ we use the definition of the bottleneck ratio  $\Phi$ to obtain:
\[\Phi \leq 
\frac{Q(S, S^c)}{\pi(S)} = \frac{\sum_{\substack{x, y \in X}} Q(x,y) \mathbf{1}[\psi(x) > t\geq \psi(y)]}{\pi(\psi > t)}
\]
Rearranging and noting that $\psi(x) > \psi(y)$ only for $x < y$, we get:
\[\pi(\psi > t) \leq \Phi^{-1} \sum_{\substack{x, y \in X}} Q(x,y) \mathbf{1}[\psi(x) > t\geq \psi(y)]\]

Integrating over $t$, noting that
\[
\int_0^\infty  \mathbf{1}[\psi(x) > t\geq \psi(y)] \, dt = \psi(x) - \psi(y),
\]
and that $ \E_\pi[\psi] = \int_{0}^\infty \pi(\psi > t),$ we obtain:
\[
\E_\pi[\psi] \leq \Phi^{-1} \sum_{x < y} (\psi(x) - \psi(y)) Q(x,y).\qedhere
\]
\end{proof}
Define the inner product \(\langle \cdot, \cdot \rangle_\pi\) on the space of real-valued functions \(f : X \to \mathbb{R}\) by
\[
\langle f, g \rangle_\pi := \sum_{x \in X} \pi(x) f(x) g(x).
\]
We view \(P\) as a matrix with entries \(P[x, y] = P[x \to y]\), and set \(P[x \to x] := -\sum_{y \neq x} P[x \to y]\), so that each row sums to zero. Note that the Dirichlet form can be written as \(\mathcal{E}(f) = \langle -P f, f \rangle_\pi\).

Since \(P\) is self-adjoint with respect to \(\langle \cdot, \cdot \rangle_\pi\), the operator \(-P\) has real, nonnegative eigenvalues. Its smallest eigenvalue is \(0\), and the spectral gap \(\lambda\) is the second smallest eigenvalue. Let \(f_2 : X \to \mathbb{R}\) be an eigenfunction corresponding to \(\lambda\), so that \(-P f_2 = \lambda f_2\). Without loss of generality, assume \(\pi(f_2 > 0) \leq 1/2\); otherwise, we may replace \(f_2\) with \(-f_2\). Define a nonnegative function \(f : X \to \mathbb{R}_{\geq 0}\) by
\[
f(x) := \max(f_2(x), 0),
\]
and let \(g := -P f\).

We show that:
\begin{equation}\label{eq:cheeger helper}
    \forall x: g(x)\leq \lambda f(x)
\end{equation}
There are two cases, $ f(x) = 0$ and $f(x) >0.$ 

If $f(x)=0:$  $ g(x) = -\sum_{y:y\neq x} P[x\to y] f(y)\leq 0 = f(x).$ 

If $ f(x) >0:$ In this case, $f(x) = f_2(x)$ and since $f(y) \geq f_2(y) \forall y,$ we have
\begin{align*}
  g(x) &= -P[x\to x] f(x) -\sum_{y:y\neq x} P[x\to y] f(y) = -P[x\to x] f_2(x) -\sum_{y:y\neq x} P[x\to y] f(y) \\
  &\leq  -P[x\to x] f_2(x) -\sum_{y:y\neq x} P[x\to y] f_2(y) 
  = (-P f_2) (x) = \lambda f_2(x) = \lambda f(x).  
\end{align*}
Since $f(x) \geq 0\forall x,$ \cref{eq:cheeger helper} yields
\begin{equation}\label{eq:cheeger proof helper equation 2}
     \mathcal{E}(f) = \langle -P f, f\rangle_\pi = \sum_{x} \pi(x)) g(x)  f(x) \leq \lambda \sum_{x} \pi(x) f(x)^2 = \lambda \langle f, f\rangle_\pi.
\end{equation}
Applying \cref{claim:cheeger helper claim} for the function $\psi=f^2,$ and use the ordering of $X$ induced $\psi,$ we get:
\[\langle f, f\rangle_\pi  =\E_\pi[\psi] \leq \Phi^{-1} \sum_{x<y} (f^2(x) - f^2(y)) Q(x,y) \]
Since both sides are nonnegative, we may square the inequality and then apply the Cauchy–Schwarz inequality to get:
\[\langle f, f\rangle_\pi^2 \leq \Phi^{-2} \left[\sum_{x<y} (f^2(x) - f^2(y)) Q(x,y) \right]^2\leq \Phi^{-2} \left[\sum_{x<y} (f(x) - f(y))^2 Q(x,y) \right] \left[\sum_{x<y} (f(x) + f(y))^2 Q(x,y) \right] \]
Using the identities $(f(x) + f(y))^2  = 2(f(x)^2 + f(y)^2) - (f(x)-f(y))^2 $, $\sum_{x<y} (f(x) - f(y))^2 Q(x,y)  = \mathcal{E}(f) ,$ and $Q(x,y)=Q(y,x),$ we can rewrite
\begin{align*}
    \sum_{x<y} (f(x) + f(y))^2 Q(x,y)  &=  2\sum_{x<y} (f(x)^2 + f(y)^2 )\pi(x) P[x\to y] -\sum_{x<y} (f(x) - f(y))^2 Q(x,y) \\
    &=2 \sum_{x\in X}\left(\pi(x) f(x)^2 \sum_{y\neq x}P[x\to y] \right) - \mathcal{E}(f)\\
    &\leq 2 \Lambda \langle f, f\rangle_\pi - \mathcal{E}(f)
\end{align*}
thus
\[\langle f, f\rangle_\pi^2 \leq \Phi^{-2} \mathcal{E}(f) (  2\Lambda \langle f, f\rangle_\pi - \mathcal{E}(f)) \leq 2\Lambda \Phi^{-2}  \mathcal{E}(f) \,   \langle f, f\rangle_\pi \]
Dividing both sides by $\langle f, f\rangle_\pi \geq 0$
and applying \cref{eq:cheeger proof helper equation 2}, we obtain
\[ \Phi^2 \leq 2 \Lambda \lambda.\]
\end{proof}